\documentclass[12pt,a4paper]{article}
\usepackage[sectionbib]{natbib}
\usepackage[margin=1in]{geometry}
\usepackage{amsfonts}
\usepackage{amsmath}
\usepackage{amsthm}
\usepackage{color}
\usepackage{mathtools}
\usepackage{algorithm}
\usepackage{sectsty}
\usepackage[colorlinks]{hyperref}
\usepackage{cleveref}
\usepackage{amssymb}
\usepackage[toc,page]{appendix}
\usepackage[mathscr]{euscript}
\usepackage[toc]{appendix}

\let\counterwithin\relax
\usepackage{chngcntr}
\usepackage{apptools}
\usepackage{booktabs}
\usepackage{lscape}
\usepackage{arydshln}
\usepackage{soul}
\usepackage{titling}
\usepackage{euscript}
\usepackage{chngcntr}
\usepackage{apptools}
\AtAppendix{\counterwithin{lemma}{section}}

\usepackage{multirow}
\usepackage{makecell}

\usepackage{graphicx}
\usepackage{subfig}
\graphicspath{{figures/}} 

\newtheorem{assumption}{Assumption}
\newtheorem{example}{Example}
\newtheorem{definition}{Definition}

\newtheorem{lemma}{Lemma}

\newtheorem{theorem}{Theorem}
\newtheorem{corollary}{Corollary}

\hypersetup{
	colorlinks=true,
	linkcolor=blue,
	citecolor=blue,
	filecolor=magenta,
	urlcolor=cyan,
}

\newcommand{\bm}{\boldsymbol}

\newcommand{\bX}{\mathbf{X}}
\newcommand{\bbeta}{\boldsymbol{\beta}}
\newcommand{\btheta}{\boldsymbol{\theta}}
\newcommand{\bgamma}{\boldsymbol{\gamma}}

\numberwithin{equation}{section}

\usepackage{authblk}

\begin{document}
\title{Quantile index regression}
\author{Yingying Zhang${}^a$, Yuefeng Si${}^b$, Guodong Li${}^b$ and Chil-Ling Tsai${}^c$
	\\ \textit{${}^a$East China Normal University, ${}^b$University of Hong Kong}
	\\ \textit{and ${}^c$University of California at Davis}}

\maketitle
	
\setlength{\parindent}{16pt}	
\setlength{\droptitle}{-6em}

\begin{abstract}
Estimating the structures at high or low quantiles has become an important subject  and attracted increasing attention across numerous fields.
However, due to data sparsity at tails, it usually is a challenging task to obtain reliable estimation, especially for high-dimensional data.
This paper suggests a flexible parametric structure to tails, and this enables us to conduct the estimation at quantile levels with rich observations and then to extrapolate the fitted structures to far tails.
The proposed model depends on some quantile indices and hence is called the quantile index regression.
Moreover, the composite quantile regression method is employed to obtain non-crossing quantile estimators, and this paper further establishes their theoretical properties, including asymptotic normality for the case with low-dimensional covariates and non-asymptotic error bounds for that with high-dimensional covariates. Simulation studies
and an empirical example are presented to illustrate the usefulness of the new model.

\end{abstract}

\textit{Keywords}: Asymptotic normality; High-dimensional analysis; Non-asymptotic property; Partially parametric model; Quantile regression.

\newpage
\section{Introduction}
Quantile regression proposed by  \cite{koenker1978regression} has been widely used across various fields such as biological science, ecology, economics, finance, and machine learning, etc.;
see, e.g., \cite{cade2003gentle}, \cite{yu2003quantile}, \cite{meinshausen2006quantile}, \cite{linton2017quantile} and
\cite{koenker2017review}.
More references on quantile
regression can be found in the books of \cite{Koenker2005} and \cite{davino2014}.
Quantile regression has also been studied for high-dimensional data; see, e.g., \cite{belloni2011}, \cite{wang2012quantile} and \cite{zheng2015globally}. On the other hand, due to practical needs, it is increasingly becoming a popular subject to estimate the structures at high or low quantiles, such as the risk of high loss for investments in finance \citep{Kuester_Mittnik_Paolella2006,Zheng_Zhu_Li_Xiao2018}, high tropical cyclone intensity and extreme waves  in climatology \citep{Elsner_Kossin_Jagger2008,Jagger_Elsner2008,lobeto2021future}, and low infant birth weights in medicine \citep{Abrevaya2001,chernozhukov2020quantile}. 
It hence is natural to make inference at these extreme quantiles for high-dimensional data, while this is still an open problem.

There are two types of approaches in the literature to model the structures at tails.
The first one is based on the conditional distribution function (CDF) of the response $Y$ for a given set of covariates $\bm{X}$, and it is usually assumed to have a semiparametric structure at tails; see, e.g., Pareto-type structures in \cite{Beirlant_Goegebeur2004} and \cite{wang2009tail}.
While this method cannot provide conditional quantiles in explicit forms.
Later, \cite{noufaily2013parametric} considered a full parametric form, the generalized gamma distribution, to the CDF and then inverted the fitted distribution into a conditional quantile distribution.
However, as indicated in \cite{racine2017nonparametric}, indirect inverse-CDF-based estimators may not be efficient in tail regions when the data has unbounded support.

The second approach is extremal quantile regression, which combines quantile regression with extreme value
theory to estimate the conditional quantile at a very high or low level of $\tau_n$, which satisfies
$(1-\tau_n)=O(n)$ with $n$ being the sample size; see \cite{chernozhukov2005extremal}.
Specifically, this is a two-stage approach: (i.) performing the estimation at intermediate quantiles $\tau_n^*$
with $(1-\tau_n^*)^{-1}=o(n)$; and (ii.) extrapolating the fitted quantile structures to those at extreme quantiles by assuming the extreme value index that is associated to the tails of conditional distributions; see \cite{wang2012estimation} and \cite{wang2013estimation} for details.
The key of this method is to make use of the feasible estimation at intermediate levels since there are relatively more observations around these levels. 
However, intermediate quantiles are also at the far tails, and the corresponding data points may still not be rich enough for the case with many covariates.

In order to handle the case with high-dimensional covariates, along the lines of extremal quantile regression, this paper suggests to conduct estimation at quantile levels with much richer observations, say some fixed levels around $\tau_0$, and then extrapolate the estimated results to extreme quantiles by fully or partially assuming a form of conditional distribution or quantile functions on $[\tau_0,1]$.
Note that there exist many quantile functions, which have explicit forms, such as the generalized lambda and Burr XII distributions \citep{Gilchrist2000}. Especially the generalized lambda distribution can provide a very accurate approximation to some Pareto-type and extreme value distributions, as well as some commonly used distributions such as Gaussian distribution \citep{Vasicek1976,Gilchrist2000}.
These flexible parametric forms can be assumed to the quantile function on $[\tau_0,1]$, and the drawback of inverting a distribution function hence can be avoided.

Specifically, for a predetermined interval ${\mathcal I}\subset[0,1]$, the quantile function of response $Y$ is assumed to have an explicit form of $Q(\tau,\bm{\theta})$ for each level $\tau\in \mathcal I$, up to unknown parameters or indices $\bm{\theta}$. By further letting $\bm{\theta}$ be a function of covariates $\bX$, we then can define the conditional quantile function as follows:
\begin{equation}\label{Model-1}	
	Q_{Y}(\tau|\bm X) =\inf\{y: F_{Y}(y|\bm X)\ge\tau\} =Q(\tau,\btheta(\bm X)), \hspace{5mm}\tau\in \mathcal{I},
\end{equation}
where $F_{Y}(\cdot|\bm X)$ is the distribution of $Y$ conditional on $\bX$,  and $\bm{\theta}(\bX)$ is a $d$-dimensional parametric function.
Note that $\bm{\theta}(\bX)$ can be referred to $d$ indices, and model \eqref{Model-1} can then be called the quantile index regression (QIR) for simplicity.
In practice, to handle high quantiles, we may take $\mathcal I=[\tau_0,1]$ with a fixed value of $\tau_0$ and then conduct a composite quantile regression (CQR) estimation for model \eqref{Model-1} at levels within $\mathcal {I}$ but with richer observations.
Subsequently, the fitted QIR model can be used to predict extreme quantiles. 
More importantly, since the estimation is conducted at fixed quantile levels, there is no difficulty to handle the case with high-dimensional covariates.
In addition, comparing with the aforementioned two types of approaches in the literature, the proposed method can not only estimate quantile regression functions effectively, but also forecast extreme quantiles directly. Finally, the QIR model naturally yields non-crossing quantile regression estimators since its quantile function is nondecreasing with respect to $\tau$.

The proposed model is introduced in details at Section 2, and the three main contributions can be summarized below:
\begin{itemize}
	\item [(a)] When conducting the CQR estimation, we encounter the first challenge on model identification, and this problem has been carefully studied for the flexible Tukey lambda distribution in Section 2.2.
	\item [(b)] Section 2.2 also derives the asymptotic normality of CQR estimators for the case with low-dimensional covariates. This is a challenging task since the corresponding objective function is non-convex and non-differentiable, and we overcome the difficulty by adopting the bracketing method in \cite{pollard1985new}.
	\item [(c)] Section 2.3 establishes non-asymptotic properties of a regularized high-dimensional estimation. This is also not trivial due to the problem at (b).
\end{itemize}

The rest of this paper is organized as follows. 
Section 3 discusses some implementation issues in searching for these estimators. Numerical studies, including simulation experiments and a real analysis, are given in Sections 4 and 5, and Section 6 provides a short conclusion and discussion. All technical details are relegated to the Appendix.

For the sake of convenience, this paper denotes vectors and matrices by boldface letters, e.g., $\mathbf{X}$ and $\mathbf{Y}$, and denotes scalars by regular letters, e.g., $X$ and $Y$. In addition, for any two real-valued sequences $\{a_n\}$ and $\{b_n\}$, denote $a_n \gtrsim b_n$ (or $a_n\lesssim b_n$) if there exists a constant $c$ such that $a_n \geq c b_n$ (or $a_n \leq c b_n$) for all $n$, and denote $a_n \asymp b_n$ if $a_n \gtrsim b_n$ and $a_n \lesssim b_n$.
For a generic vector $\bX$ and matrix $\mathbf{Y}$, let $\|\bX\|$, $\|\bX\|_1$ and $\|\mathbf{Y}\|_{\textrm{F}}$ represent the Euclidean norm, $\ell_1$-norm and Frobenius norm, respectively.

\section{Quantile index regression}\label{sec:2}

\subsection{Quantile index regression model}\label{sec:qir}
Consider a response $Y$ and a $p$-dimensional vector of covariates $\bm X=(X_1,...,X_p)^{\prime}$.
We then rewrite the quantile function of $Y$ conditional on $\bm X$ at \eqref{Model-1} with an explicit form of $\btheta(\bm X,\bm \beta)$ below,
\begin{equation}\label{Model2}
Q_{Y}(\tau|\bm X) =Q(\tau,\btheta(\bm X,\bm \beta)), \hspace{5mm}\tau\in\mathcal{I},
\end{equation}
where $\mathcal{I}\subset[0,1]$ is an interval or the union of multiple disjoint intervals, 
$\btheta(\bm X,\bm \beta)=(\theta_1(\bm X,\bm \beta),\cdots,\theta_d(\bm X,\bm \beta))^\prime$, $\bm\beta=(\bm\beta_1^\prime,...,\bm\beta_d^\prime)^{\prime}$,
$\bm\beta_j=(\beta_{j1},\cdots,\beta_{jp})^{\prime}$,
$\theta_j(\bm X,\bm \beta)=g_j(\bm X^\prime \bm\beta_j)$, the link functions $g_j^{-1}(\cdot)$s are monotonic for $1\leq j\leq d$,   and the intercept term can be included by letting $X_1=1$.
We call model \eqref{Model2} the quantile index regression (QIR) for simplicity, and the two examples of $Q(\cdot,\cdot)$ below are first introduced to illustrate the new model.

\begin{example}\label{example2}
	Consider the location shift model, $Q(\tau, \theta)=\theta+Q_{\Phi}(\tau)$, for all $\tau\in[\tau_0,1]$, where $\tau_0\in(0,1)$ is a fixed level, $\theta\in\mathbb{R}$ is the location index and $Q_{\Phi}(\tau)$ is the quantile function of standard normality.
	Under the identity link function, $\theta(\bX,\bbeta)=\bX'\bbeta$, we can construct a linear quantile regression model at $\tau_0$. Then, after estimation, we can make a prediction at any level of $\tau\in(\tau_0,1)$.
\end{example}

\begin{example}\label{TukeyLambda}
	Consider the Tukey lambda distribution \citep{Vasicek1976} that is defined by its quantile function,
	\[
	Q(\tau,\btheta)=\theta_{1}+{\theta_{2}}\frac{\tau^{\theta_{3}}-(1-\tau)^{\theta_{3}}}{\theta_{3}},
	\]
	where $\btheta=(\theta_{1},\theta_2,\theta_{3})'$, and $\theta_1\in\mathbb{R}$, $\theta_2>0$ and $\theta_3\le 1$ are the location, scale and tail indices, respectively.
Due to its flexibility, the Tukey lambda distribution can provide an accurate approximation to many commonly used distributions, such as normal, logistic, Weibull, uniform, and Cauchy distributions, etc.; see \cite{Gilchrist2000}.
It is then expected to have a better performance when model \eqref{Model2} is combined with this distribution.
\end{example}


In the literature, there exist many  quantile functions, which have explicit forms, such as the generalized lambda and Burr XII distributions \citep{Gilchrist2000,fournier2007estimating}.
For example, the generalized lambda distribution has the form of
\[
Q(\tau,\btheta)=\theta_{1}+{\theta_{2}}\left( \frac{\tau^{\theta_{3}}-1}{\theta_{3}}-\frac{(1-\tau)^{\theta_{4}}-1}{\theta_{4}}\right),
\]
where the indices $\theta_3$ and $\theta_4$ control the right and left tails, respectively. Note that  it reduces to the Tukey lambda distribution when $\theta_3=\theta_4$.
This indicates that the generalized lambda distribution can be considered if we focus on the quantiles with the full range, i.e. $\mathcal{I}=[0,1]$. On the other hand, the Tukey lambda may be a better choice if our interest is on the quantiles at one side only, i.e. $\mathcal{I}\subset[0,0.5)$ or $\mathcal{I}\subset(0.5,1]$.

\subsection{Low-dimensional composite quantile regression estimation} \label{sec:model-ld}

Denote the observed data by $\{(Y_i,\bm X_i^{\prime})^{\prime},i=1,...,n\}$, and they are independent and identically
distributed ($i.i.d.$) samples of random vector $(Y,\bm X)$, where $Y_i$ is the response,  $\bm X_i=(X_{i1},...,X_{ip})^\prime$ contains $p$ covariates, and $n$ is the number of observations.

Let $\tau_{1}\le\tau_{2}\le\ldots\le\tau_{K}$ be $K$ fixed quantile levels, where $\tau_{k}\in\mathcal{I}$ for all $1\le k\le K$.
To achieve higher efficiency, we consider the composite quantile regression (CQR) estimator below.
\begin{equation}\label{eq:CQR}
\widehat{\bbeta}_n=\arg\min_{\bbeta}L_n(\bbeta) \hspace{2mm}\text{and}\hspace{2mm} L_n(\bbeta)=\sum_{k=1}^{K}\sum_{i=1}^{n}\rho_{\tau_{k}}\{Y_{i}-Q(\tau_{k},\btheta(\bX_{i},\bbeta))\},
\end{equation}
where $\rho_{\tau}(x)=x\{\tau-I(x<0)\}$ is the quantile check function; see \cite{Zou_Yuan2008} and \cite{kai2010local}. 
Note that  $Q(\tau,\btheta(\bX,\widehat{\bbeta}_{n}))$ has the non-crossing property with respect to $\tau$  since, for each $\btheta$, $Q(\tau,\btheta)$ is a non-decreasing quantile function.

To study the asymptotic properties of $\widehat{\bbeta}_n$, we first consider $\bm\beta_0=(\bm\beta_{01}^\prime,...,\bm\beta_{0d}^\prime)^{\prime}=\arg\min_{\bbeta}\bar{L}(\bbeta)$, where $ \bar{L}(\bbeta)=E[\sum_{k=1}^{K}\rho_{\tau_{k}}\{Y-Q(\tau_{k},\btheta(\bX,\bbeta))\}]$ is the population loss function, and it may not be unique, i.e. the CQR estimation at \eqref{eq:CQR} may suffer from the identification problem.
For the sake of illustration, let us
consider the  case without including covariates $\bX$, i.e. we estimate $\btheta$
with a sequence of quantile  levels $\tau_{k}\in\mathcal{I}$ and  $1\le k\le K$.
To this end, it requires that two different values of $\btheta$ can not yield the same quantile function $Q(\tau_{k},\btheta)$ across all $K$  levels. In other words, if there exists $\btheta\ne \btheta^{*}$ that yield $Q(\tau_{k},\btheta)=Q(\tau_{k},\btheta^{*})$ for all $K$ quantiles, then
$\btheta$ and $\btheta^{*}$ are not identifiable.
In sum, to guarantee that $\bbeta_{0}$ is the unique minimizer of the population loss, we make the following assumption on the quantile function $Q(\tau,\btheta)$.

\begin{assumption}\label{assum:identification}
For any two index vectors $\bm\theta_1\neq \bm\theta_2$, there exists at least one $1\le k\le K$ such that   $Q(\tau_k,\bm\theta_1)\neq Q(\tau_k,\bm\theta_2)$.
\end{assumption}

Intuitively, for any quantile function $Q(\tau,\btheta)$, one can always increase the number of quantile levels $K$ to make Assumption \ref{assum:identification} hold. However, it may also depend on the structures of quantile functions and the number of indices.
For the sake of illustration, we state the sufficient and necessary condition for the Tukey lambda distribution that satisfies Assumption \ref{assum:identification}.

\begin{lemma}\label{lemma:tukey}
For the Tukey lambda distribution defined in Example \ref{TukeyLambda}, we have that
 (i) for $\tau_{k}\in(0,1)$ with $1\le k\le K$, Assumption \ref{assum:identification} holds if $K\ge 4$;
(ii) for $\tau_{k}\in(0.5,1)$ (or $\tau_{k}\in(0,0.5)$) with $1\le k\le K$, Assumption \ref{assum:identification} holds if and only if $K\ge 3$.
\end{lemma}

Assumption \ref{assum:identification}, together with an additional assumption on covariates $\bX$, allows us to show that $\bbeta_{0}$ is the unique minimizer of $\bar{L}(\bbeta)$; see the following theorem, which is critical to establish the asymptotic properties of $\widehat{\bbeta}_{n}$.

\begin{theorem}\label{thm-identification}
	Suppose that $E(\bm X\bm X^{\prime})$ is finite and positive definite. If Assumption \ref{assum:identification} holds, then $\bm\beta_0$ is the unique minimizer of $\bar{L}(\bbeta)$.
\end{theorem}

To demonstrate the consistency of $\widehat{\bbeta}_{n}$ given below, we  assume that the parameter space $\Theta\subset\mathbb{R}^{dp}$ is compact, and the true parameter vector $\bm\beta_0$ is an interior point of $\Theta$.

\begin{theorem}\label{theorem1}
	Suppose that
	$ E\{\max_{1\leq k\leq K}\sup_{\bbeta\in\Theta}\|\partial Q(\tau_{k},\btheta(\bX,\bbeta))/\partial\bbeta\|\}<\infty. $
	If the conditions of Theorem \ref{thm-identification} hold, then $\widehat{\bbeta}_n\rightarrow \bbeta_{0}$ in probability as $n\rightarrow \infty$.
\end{theorem}
\noindent
The moment condition assumed in the above theorem allows us to adopt the uniform consistency theorem of  \cite{andrews1987consistency} in our technical proofs. To show the asymptotic distribution of $\widehat{\bbeta}_{n}$, we introduce two additional  assumptions given below.

\begin{assumption}\label{assum:EE}
	For all $1\le k\le K$,
	$$E\left\|\frac{\partial Q(\tau_{k},\btheta(\bX,\bbeta_{0}))}{\partial\bbeta}\right\|^{3}<\infty \hspace{5mm}\text{and}\hspace{5mm}
	E\sup_{\bbeta\in\Theta}\left\|\frac{\partial^{2} Q(\tau_{k},\btheta(\bX,\bbeta))}{\partial\bbeta\partial\bbeta^{\prime}}\right\|_{\mathrm{F}}^{2}<\infty.$$
\end{assumption}

\begin{assumption}\label{assum:density}
	The conditional density $f_{Y}(y|\bm X)$ is bounded and continuous uniformly for all $\bm X$.
\end{assumption}
\noindent
Assumption \ref{assum:EE} is used to prove Lemma \ref{lem-thm3} in the Appendix; see also \cite{zhu2011global}.
Assumption \ref{assum:density} is commonly used in the literature of quantile regression  \citep{Koenker2005,Belloni_conditionalprocess}, and it can be relaxed by providing more complicated and lengthy technical details \citep{kato2012asymptotics,chernozhukov2015quantile,galvao2016smoothed}.

Denote
\[
\Omega_{0}=\sum_{k'=1}^{K}\sum_{k=1}^{K}\min\{\tau_{k},\tau_{k'}\}\left(1-\max\{\tau_{k},\tau_{k'}\}\right)E\left[\frac{\partial Q(\tau_{k},\btheta(\bX,\bbeta_{0}))}{\partial\bbeta}\frac{\partial Q(\tau_{k'},\btheta(\bX,\bbeta_{0}))}{\partial\bbeta^{\prime}}\right]
\]
and
\[
\Omega_{1}=\sum_{k=1}^{K}E\left[f_{Y}\left\{Q(\tau_{k},\btheta(\bX,\bbeta_{0}))|\bX\right\}\frac{\partial Q(\tau_{k},\btheta(\bX,\bbeta_{0}))}{\partial\bbeta}\frac{\partial Q(\tau_{k},\btheta(\bX,\bbeta_{0}))}{\partial\bbeta^{\prime}}\right].
\]

\begin{theorem}\label{theorem2}
	Suppose that Assumptions \ref{assum:EE} and \ref{assum:density} hold, and $\Omega_{1}$ is positive definite. If the conditions of Theorem \ref{theorem1} are satisfied, then $\sqrt{n}(\widehat{\bbeta}_n-\bbeta_{0})\rightarrow N(\bm{0},\Omega_{1}^{-1}\Omega_{0}\Omega_{1}^{-1})$ in distribution as $n\rightarrow\infty$.
\end{theorem}

Note that the objective function $L_n(\bbeta)$ is non-convex and non-differentiable, and this makes it challenging to establish the asymptotic normality of $\widehat{\bbeta}_n$.
We overcome the difficulty by making use of the bracketing method in \cite{pollard1985new}. 
Moreover, to estimate the asymptotic variance in Theorem \ref{theorem2}, we first apply the nonparametric method in
\cite{Hendricks_Koenker1991} to estimate the quantities of $f_{Y}\left\{Q(\tau_{k},\btheta(\bX,\bbeta_{0}))|\bX\right\}$ with $1\leq k\leq K$, and then matrices $\Omega_{0}$ and $\Omega_{1}$ can be approximated by the sample averaging.

In addition, based on  the estimator $\widehat{\bbeta}_{n}$,  one can use $Q(\tau^{*},\btheta(\bX,\widehat{\bbeta}_{n}))$ to further predict the conditional quantile at any level $\tau^{*}\in\mathcal{I}$, and the corresponding theoretical justification can be established by directly applying the delta-method \citep[Chapter 3]{vanderVaart1998}. 
\begin{corollary}
 	Suppose that the conditions of Theorem \ref{theorem2} are satisfied. Then, for any $\tau^*\in\mathcal{I}$,
 	$$\sqrt{n}\{Q(\tau^*,\btheta(\bX,\widehat{\bbeta}_{n}))-Q(\tau^*,\btheta(\bX,\bbeta_0))\}\rightarrow N(\bm{0}, \bm{\delta}^{\prime}\Omega_{1}^{-1}\Omega_{0}\Omega_{1}^{-1}\bm{\delta})$$ in distribution as $n\rightarrow\infty$, where $\bm{\delta}=E[\partial Q(\tau^*,\btheta(\bX,\bbeta_{0}))/\partial\bbeta]\in\mathbb{R}^{dp}$.
\end{corollary}

\subsection{High-dimensional regularized estimation}\label{sec:2.3}

This subsection considers the case with high-dimensional covariates, i.e., $p\gg n$, and the true parameter vector $\bbeta_{0}$ is assumed to be $s$-sparse, i.e. the number of nonzero elements in $\bbeta_{0}$ is no more than $s>0$.
A regularized CQR estimation can then be introduced,
\begin{equation}\label{program2}
	\widetilde{\bbeta}_n= \arg\min_{\bbeta\in\Theta} n^{-1}L_n(\bbeta)+\sum_{j=1}^dp_{\lambda}(\bbeta_j),
\end{equation}
where $\Theta$ is given in Theorem \ref{high-dim}, $p_{\lambda}$  is
a penalty function, and it
depends on a tuning (regularization) parameter $\lambda\in \mathbb{R}^+$ with $\mathbb{R}^+=(0,\infty)$.

Consider the loss function  $L_n(\bbeta)=\sum_{k=1}^{K}\sum_{i=1}^{n}\rho_{\tau_{k}}\{Y_{i}-Q(\tau_{k},\btheta(\bX_{i},\bbeta))\}$ defined in (\ref{eq:CQR}), and $Q(\tau_{k},\btheta(\bX_{i},\bbeta))$ usually is nonconvex with respect to $\bbeta$. As a result,
$L_n(\bbeta)$ will be nonconvex although the check loss $\rho_{\tau}(\cdot)$ is convex, and there is no more harm to use nonconvex penalty functions. Specifically, we consider the component-wise penalization,
\[
\sum_{j=1}^dp_{\lambda}(\bbeta_j)=\sum_{j=1}^d\sum_{l=1}^{p}p_{\lambda}(\beta_{jl}),
\]
where $p_{\lambda}(\cdot)$ is possibly nonconvex and satisfies the following assumption.

\begin{assumption}\label{regular}
	The univariate function $p_{\lambda}(\cdot)$ satisfies the following conditions: (i) it is symmetric around zero with $p_{\lambda}(0)=0$; (ii) it is is nondecreasing on the nonnegative real line; (iii) the function $p_{\lambda}(t)/t$ is nonincreasing with respect to $t\in \mathbb{R}^+$; (iv) it is differentiable for all $t\neq0$ and subdifferentiable at $t=0$, with $\lim_{t\rightarrow 0^{+}}p_{\lambda}^{\prime}(t)=\lambda L$ and $L$ being a constant; (v) there exists $\mu>0$ such that $p_{\lambda,\mu}=p_{\lambda}(t)+\frac{\mu^{2}}{2}t^{2}$ is convex.
\end{assumption}

The above is the $\mu$-amenable assumption given in \cite{loh2015regularized} and \cite{loh2017statistical},
and the penalty function is required not too far from the convexity.
Note that the popular penalty functions, including SCAD \citep{Fan_Li2001} and
MCP \citep{Zhang2010}, satisfy the above properties.

In the literature of nonconvex penalized quantile regression, \cite{Jiang_Jiang_Song2012} studied nonlinear quantile regressions with SCAD regularizer from the asymptotic viewpoint, while it can only handle the case with $p=o(n^{1/3})$. 
\cite{wang2012quantile} and \cite{sherwood2016partially} considered the case that $p$ grows exponentially with $n$, and their proving techniques heavily depend on the condition that the loss function should be represented as a difference of the two convex functions. However, $L_n(\bbeta)$ does not meet this requirement since quantile function $Q(\tau,\btheta)$ can be nonconvex.

On the other hand, non-asymptotic properties recently have attracted considerable attention in the theories of high-dimensional analysis; see, e.g., \cite{belloni2011,sivakumar2017high,pan2020multiplier}.
This subsection attempts to study them for our proposed quantile estimators, while it is a nontrival task since existing results only focused on linear quantile regression.
\cite{loh2015regularized} and \cite{loh2017statistical} studied the non-asymptotic properties for M-estimators with both nonconvex loss and regularizers, while they required the loss function to be twice differentiable.
The technical proofs in the Appendix follow the framework in \cite{loh2015regularized} and \cite{loh2017statistical}, and some new techniques are developed to tackle the nondifferentiability of the quantile check function.

Let $\btheta(\bgamma)=(g_1(\gamma_1),\ldots,g_d(\gamma_d))$ with $g_j^{-1}(\cdot)$s being link functions and $\bgamma=(\gamma_1,\ldots,\gamma_d)$, and we can then denote $Q(\tau,\bgamma):=Q(\tau,\btheta(\bgamma))$.
Moreover, by letting $\gamma_j(\bX,\bbeta)=\bX\bbeta_{j}$ for $1\leq j\leq d$ and $\bgamma(\bX,\bbeta)=(\gamma_1(\bX,\bbeta),\ldots,\gamma_d(\bX,\bbeta))$, we can further denote  $Q(\tau,\bgamma(\bX,\bbeta)):=Q(\tau,\btheta(\bX,\bbeta))$.

\begin{assumption}\label{Cond-qf}
	Quantile function $Q(\tau,\bgamma)$ is differentiable with respect to $\bgamma$, and there exist two positive constants $L_Q$ and $C_{X}$ such that $\max_{1\leq k\leq K}\|{\partial Q(\tau_{k},\bgamma)}/{\partial\bgamma}\|\leq L_Q$ and $\|\bX\|_{\infty}\le C_{X}$.
\end{assumption}

The differentiable assumption of quantile functions allows us to use the Lipschitz property and multivariate contraction theorem.
The boundedness of covariates is to assure that the bounded difference inequality can be used, and it can be relaxed with more complicated and lengthy technical details \citep{Wang_He_2021}.

Denote by $\mathcal{B}_{R}(\bbeta_0)=\{\bbeta\in\mathbb{R}^{dp}: \|\bbeta-\bbeta_0\|\leq R\}$ the Euclidean ball centered at $\bbeta_0$ with radius $R>0$, and let $\lambda_{\min}(\bbeta)$ be the smallest eigenvalue of matrix
\[
\Omega_{2}(\bbeta)=\sum_{k=1}^{K}E\left[\frac{\partial Q(\tau_{k},\btheta(\bX,\bbeta))}{\partial\bbeta}\frac{\partial Q(\tau_{k},\btheta(\bX,\bbeta))}{\partial\bbeta^{\prime}}\right].
\]

\begin{assumption}\label{Cbeta}
	There exists a fixed $R>0$ such that
	$\lambda_{\min}^0 = \inf_{\bbeta\in\mathcal{B}_{R}(\bbeta_{0})}\lambda_{\min}(\bbeta)>0$, and assume that $f_{\min}=\min_{1\leq k\leq K}\inf_{\bbeta\in\mathcal{B}_{R}(\bbeta_{0})}f_{Y}\left\{Q(\tau_{k},\btheta(\bX,\bbeta))|\bX\right\}>0$.
\end{assumption}

The above assumption guarantees that the population loss $\bar{L}(\bbeta)=E[n^{-1}L_n(\bbeta)]$ is strongly convex around the true parameter vector $\bbeta_{0}$. Specifically, let $\bar{\mathcal{E}}(\Delta) =\bar{L}(\bbeta_0+\Delta)-\bar{L}(\bbeta_0)-\Delta^{\prime}{\partial \bar{L}(\bbeta_0) }/{\partial \bbeta}$ be the first-order Taylor expansion. Then, by Assumption \ref{Cbeta}, we have that $\bar{\mathcal{E}}(\Delta) \geq 0.5f_{\min}\lambda_{\min}^0\|\Delta\|^2$ for all $\Delta$ such that $\|\Delta\|\leq R$; see Lemma \ref{RSC} in the Appendix
for details. We next obtain the non-asymptotic estimation bound of $\widetilde{\bbeta}_n$.

\begin{theorem}\label{high-dim}
	Suppose that Assumptions  \ref{assum:identification} and \ref{regular}-\ref{Cbeta} hold, $2f_{\min}\lambda_{\min}^{0}>\mu $, $n\gtrsim\log p$ and $\lambda \gtrsim  \sqrt{\log p/n}$.
	Then the minimizer $\widetilde{\bbeta}_n$ of \eqref{program2} with $\Theta=\mathcal{B}_{R}(\bbeta_{0})$ satisfies the error bounds of
	\begin{align*}
		\|\widetilde{\bbeta}_n-\bbeta_{0}\|\le\frac{6L\sqrt{s}\lambda}{4\alpha-\mu}~~~~~\mbox{and}~~~~~\|\widetilde{\bbeta}_n-\bbeta_{0}\|_{1}\le\frac{24Ls\lambda}{4\alpha-\mu},
	\end{align*}
	with probability at least $1-c_{1}p^{-c_{2}}-K\max\{\log p, \log n\} p^{-c^{2}}$ for any $c>1$, where $\alpha=0.5f_{\min}\lambda_{\min}^0$, $\mu$ and $L$ are defined in Assumption \ref{regular}, $s$ is the number of nonzero elements in $\bbeta_{0}$, and the constants
		$c_{1}$ and $c_{2}>0$ are given in Lemma \ref{infty} of the Appendix.
\end{theorem}

In practice, we can choose $\lambda\asymp \sqrt{\log p/n}$, and it then holds that $\|\widetilde{\bbeta}_n-\bbeta_{0}\|\lesssim \sqrt{s\log p/n}$, which has the standard rate of error bounds; see, e.g., \citep{loh2017statistical}.
Moreover, for the predicted conditional quantile of $Q(\tau^{*},\btheta(\bX,\widetilde{\bbeta}_{n}))$ at any level $\tau^{*}\in\mathcal{I}$, it can be readily verified that $|Q(\tau^{*},\btheta(\bX,\widetilde{\bbeta}_{n}))-Q(\tau^{*},\btheta(\bX,{\bbeta}_0))|$ has the same convergence rate as $\|\widetilde{\bbeta}_n-\bbeta_{0}\|$.
Finally, the above theorem requires the minimization \eqref{program2} to be conducted in $\Theta=\mathcal{B}_{R}(\bbeta_{0})$, which is unknown but fixed.
This enables us to solve the problem by conducting a random initialization in optimizing algorithms.

\section{Implementation issues}

\subsection{Optimizing algorithms}
This subsection provides algorithms to search for the CQR estimator at \eqref{eq:CQR} and regularized estimator at \eqref{program2}. 

For the CQR estimation without penalty at \eqref{eq:CQR}, we employ the commonly used gradient descent algorithm to search for estimators, and the $(r+1)$th update is given by
\[
{\bbeta}^{(r+1)} = {\bbeta}^{(r)} - \eta^{(r)} \, \nabla L_n({\bbeta}^{(r)}),
\]
where $\widehat{\bbeta}_n^{(r)}$ is from the $r$th iteration, and $ \eta ^{(r)}$ is the step size. 
Note that the quantile check loss is nondifferentiable at zero, and $\nabla L_n({\bbeta}^{(r)})$ in the above refers to the subgradient \citep{moon2021learning} instead.
In practice, too small step size will cause the algorithm to converge slowly, while too large step size may cause the algorithm to diverge.
We choose the step size by the backtracking line search (BLS) method, which is shown to be simple and effective; see \cite{nonlinearprog2016}. 
Specifically, the algorithm starts with a large step size and, at $(r+1)$th update, it is reduced by keeping multiplying a fraction of $b$ until $L_{n}(\bbeta^{(r+1)})-L_{n}(\bbeta^{(r)})<-a\eta^{(r)}\|\nabla L_{n}(\bbeta^{(r)})\|_{2}^{2}$, where $a$ is another hyper-parameter. The simulation experiments in Section 4 work well with the setting of $(a,b)=(0.3,0.5)$.

For the regularized estimation at \eqref{program2}, we adopt the composite gradient descent algorithm \citep{loh2015regularized}, which is designed for a nonconvex problem and fits our objective functions well.
Consider the SCAD penalty, which satisfies Assumption \ref{regular} with $ L=1 $ and $ \mu = 1/(\alpha -1)$.
We then can rewrite the optimization problem at \eqref{program2} into
\begin{equation*}
	\widetilde{\bbeta}_n= \arg\min_{\bbeta\in\Theta} \underbrace{\{n^{-1}L_n(\bbeta)-\mu\|\bbeta\|_{2}^{2}/2\}}_{\widetilde{L}_{n}(\bbeta)}+\lambda g(\bbeta),
\end{equation*}
where, from Assumption \ref{regular}, $g(\bbeta)=\{\sum_{j=1}^{d}p_{\lambda}(\bbeta_{j})+\mu\|\bbeta\|_{2}^{2}/2\}/\lambda$ is convex. 
As a result, similar to the composite gradient descent algorithm in \cite{loh2015regularized}, the $(r+1)$th update can be calculated by
\begin{equation*}
	{\bbeta}^{(r+1)}=\arg\min\left\{\|\bbeta-(\bbeta^{(r)}-\eta\nabla\widetilde{L}_{n}(\bbeta))\|_{2}^{2}/2+\lambda\eta g(\bbeta)\right\},
\end{equation*}
which has a closed-form solution of
\begin{align*}
	{\bbeta}^{(r+1)}=\left\{\begin{aligned}
		&0,&\quad 0 \leq |z| \leq \nu \lambda\\
		&z - \textrm{sign}(z) \cdot \nu\lambda,&\quad \nu \lambda \leq |z| \leq (\nu+1) \lambda\\
		&\{z - \textrm{sign}(z) \cdot \frac{\alpha \nu\lambda}{\alpha - 1}\}/\{1 - \frac{\nu}{\alpha -1}\},&\quad (\nu+1) \lambda \leq |z| \leq \alpha \lambda\\
		&z, &\quad |z| \geq \alpha\lambda
	\end{aligned}\right.
\end{align*}
with $z=(\bbeta^{(r)}-\eta\nabla\widetilde{L}_{n}(\bbeta^{(r)}))/(1+\mu\eta)$ and $\nu = \eta / (1 + \mu \eta)$, 
where the step size $\eta$ is chosen by the BLS method. 

\subsection{Hyper-parameter selection}

There are two types of hyper-parameters in the penalized estimation at \eqref{program2}: the tuning parameter $\lambda$ and quantile levels of $\tau_k$ with $1\leq k\leq K$.
We can employ validation methods to select the tuning parameter $\lambda$ such that the composite quantile check loss is minimized.

The selection of $\tau_k$'s is another important task since it will affect the efficiency of resulting estimators.
Suppose that we are interested in some high quantiles of $\tau_m^*$ with $1\leq m\leq M$, and then the QIR model can be assumed to the interval of $\mathcal{I}=[\tau_0,1]$, which contains all $\tau_m^*$'s.
We may further choose a suitable interval of $[\tau_{L},\tau_{U}]\subset \mathcal{I}$ such that $\tau_k$'s can be equally spaced on it, i.e. $\tau_k=\tau_L+k(\tau_U-\tau_L)/(K+1)$ for $1\leq k\leq K$, where it can be set to $\tau_0=\tau_L$.
As a result, the selection of $\tau_k$'s is equivalent to that of $[\tau_L,\tau_U]$.

We may choose $\tau_U$ such that it is close to $\tau_m^*$'s, while a reliable estimation can be afforded at this level. The selection of $\tau_L$ is a trade-off between estimation efficiency and model misspecification; see \cite{wang2012estimation,wang2009tail}.
On one hand, to improve estimation efficiency, we may choose $\tau_L$ close to 0.5 since the richest observations will appear at the middle for most real data.
On the other hand, we have to assume the parametric structure over the whole interval of $[\tau_L,1]$, i.e. more limitations will be added to the real example.
The criterion of prediction errors (PEs) is hence introduced, 
\[
PE = \frac{1}{M}\sum_{m=1}^{M}\frac{1}{\sqrt{\tau_{m}^{*}(1-\tau_{m}^{*})}} \cdot  \sqrt{n}\left|\frac{1}{n}\sum_{i=1}^n I\{y_i < \widehat{Q}_{Y}(\tau_{m}^{*} | \bX_i)\} - \tau_{m}^* \right|,
\]
where we will choose $\tau_{L}$ with the minimum value of PEs; see also \cite{wang2012estimation}.

In practice, the cross validation method can be used to select $\lambda$ and $\tau_L$ simultaneously.
Specifically, the composite quantile check loss and PEs are both evaluated at validation sets. 
For each candidate interval of $[\tau_{L},\tau_{U}]$, the tuning parameter $\lambda$ is selected according to the composite quantile check loss, and the corresponding value of PE is also recorded.
We then will choose the value of $\tau_{L}$, which corresponds to the minimum value of PEs among all candidate intervals.


\section{Simulation Studies}
\subsection{Composite quantile regression estimation} \label{sec:sim-cqr}

This subsection conducts simulation experiments to evaluate the finite-sample performance of the low-dimensional composite quantile regression (CQR) estimation at \eqref{eq:CQR}.

The Tukey Lambda distribution in Example \ref{TukeyLambda} is used to generate the $i.i.d.$ sample,
\begin{align}
	\begin{split}\label{eq:dgp}
	Y_i & = Q_Y(U_i, \btheta(\bX_i,\bbeta)) = \theta_{1}(\bX_i,\bbeta)+\theta_{2}(\bX_i,\bbeta)\frac{U_{i}^{\theta_{3}(\bX_i,\bbeta)}-(1-U_{i})^{\theta_{3}(\bX_i,\bbeta)}}{\theta_{3}(\bX_i,\bbeta)}\\
	&\theta_{1}(\bX_i,\bbeta) = g_1(\bX_i'\bbeta_1),~ \theta_{2}(\bX_i,\bbeta) = g_2(\bX_i'\bbeta_2),~ \theta_{3}(\bX_i,\bbeta) = g_3(\bX_i'\bbeta_3),
	\end{split}
\end{align}
where $\{U_i\}$ are independent and follow $\text{Uniform}(0, 1)$, $\bX_i=(1,X_{i1},X_{i2})^{\prime}$, $\{(X_{i1},X_{i2})^{\prime}\}$ is an $i.i.d.$ sequence with 2-dimensional standard normality.
The true parameter vector is $ \bbeta_0 = (\bbeta_{01}',\bbeta_{02}',\bbeta_{03}')'$, and we set the location parameters $\bbeta_{01}=(1, 0.5, -1)'$, the scale parameters $\bbeta_{02}=(1, 0.5, -1)'$ and the tail parameters $\bbeta_{03}=(1, -1, 1)' $. 
For the tail index $\theta_{3}(\bX_i,\bbeta)$, before generating the data, we first scale each covariate into the range of $ [-0.5, 0.5] $ such that a relatively stable sample can be generated.
In addition, $ g_1$, $g_2$ and $g_3 $ are the inverse of link functions.
We choose identity link for the location index and softplus-related link for the scale and tail indices, i.e., $ g_1(x) = x, g_2(x) = \textrm{softplus}(x)$ and $ g_3(x) = 1-\textrm{softplus}(x) $, where $ \textrm{softplus}(x) = \log(1+\exp(x)) $ is a smoothed version of $ x_+ = \max\{0, x\} $ and hence the name. Note that $g_{2}(x)>0$ and $g_{3}(x)<1$.
We consider three sample sizes of $n=500$, 1000 and 2000, and there are 500 replications for each sample size.

The algorithm for CQR estimation in Section 3 is applied with $K=10$ and $\tau_k$'s being equally spaced over $[\tau_L,\tau_U]$.
We consider three quantile ranges of $(\tau_L,\tau_U)=(0.5, 0.99)$, $(0.7, 0.99)$ and $(0.9, 0.99)$, and the estimation efficiency is first evaluated. 
Figure \ref{box:low} gives the boxplots of three fitted location parameters $\widehat{\bbeta}_{1n}=(\widehat{\beta}_{1,1},\widehat{\beta}_{1,2},\widehat{\beta}_{1,3})^{\prime}$.
It can be seen that both bias and standard deviation decrease as the sample size increase.
Moreover, when $\tau_L$ decreases, the quantile levels with richer observations will be used for the estimation and, as expected, both bias and standard deviation will decrease.
Boxplots for fitted scale and tail parameters show a similar pattern and hence are omitted to save the space.

We next evaluate the prediction performance of $Q(\tau^{*},\btheta(\bX,\widehat{\bbeta}_{n}))$ at two interesting quantile levels of $\tau^*=0.991$ and 0.995. Consider two values of covariates, $\bm{X}=(1,0.1,-0.2)^{\prime}$ and $(1 ,0, 0)^{\prime} $, and the corresponding tail indices are $\theta_{3}(\bX,\bbeta_0)=-0.1139$ and $-0.3132$, respectively.
Note that the Tukey lambda distribution can provide a good approximation to Cauchy and normal distributions when the tail indices are $-1$ and $0.14$, respectively, and it becomes more heavy-tailed when the tail index decreases \citep{freimer1988study}.  
The prediction error in terms of squared loss (PES), $[Q(\tau^{*},\btheta(\bX,\widehat{\bbeta}_{n}))-Q(\tau^{*},\btheta(\bX,{\bbeta}_0))]^2$, is calculated for each replication, and the corresponding sample mean refers to the commonly used mean square error.
Table \ref{tab:low} presents both the sample mean and standard deviation of PESs across 500 replications as in \cite{wang2012estimation}.
A clear trend of improvement can be observed as the sample size becomes larger, and the prediction is more accurate at the 99.1-th quantile level for almost all cases.

\subsection{High-dimensional regularized estimation} \label{sec:sim-hd}
This subsection conducts simulation experiments to evaluate the finite-sample performance of the high-dimensional regularized estimation at \eqref{program2}.

For the data generating process at \eqref{eq:dgp}, we consider three dimensions of $p=50$, $100$ and $150$, and the true parameter vectors are extended from those in Section 4.1 by adding zeros, i.e.
$ \bbeta_{01}=(1, 0.5, -1, 0, ..., 0)'$, $ \bbeta_{02}=(1, 0.5, -1, 0, ..., 0)'$ and $\bbeta_{03}=(1, -1, 1, 0, ..., 0)'$, which are vectors of length $p$ with $3$ non-zero entries. 
As a result, all true parameters $\bbeta_{0}=(\bbeta_{01}',\bbeta_{02}',\bbeta_{03}')'$ make a vector of length $3p$ with $s=9$ non-zero entries.
The sample size is chosen such that $n=\left\lfloor cs\log p\right\rfloor$ with $c=5$, 10, 20, 30, 40 and 50, where $\lfloor x\rfloor$ refers to the largest integer smaller than or equal to $x$. All other settings are the same as in the previous subsection.

The algorithm for regularized estimation in Section 3 is used to search for the estimators, and we generate an independent validation set of size $5n$ to select tuning parameter $\lambda$ by minimizing the composite quantile check loss; see also \cite{wang2012quantile}. 
Figure \ref{fig:high} gives the estimation errors of $\|\widetilde{\bbeta}_n-\bbeta_{0}\|$. It can be seen that $\|\widetilde{\bbeta}_n-\bbeta_{0}\|$ is roughly proportional to the quantity of $\sqrt{s\log p/n}$, and this confirms the convergence rate in Theorem \ref{high-dim}. Moreover, the estimation errors approach zero as the sample size $n$ increase, and we can then conclude the consistency of $\widetilde{\bbeta}_n$.
Finally, when $\tau_L$ increases, the quantile levels with less observations will be used in the estimating procedure, and hence larger estimation errors can be observed.

We next evaluate the prediction performance at quantile levels $\tau^*= 0.991 $ and $ 0.995 $, and covariates $\bX$ take values of $(1,0.1,-0.2,0,\cdots,0)'$ and $(1,0,0,0,\cdots,0)'$, similar to those in the previous subsection.
Table \ref{tab:high1} gives mean square errors of the predicted conditional quantiles $Q(\tau^{*},\btheta(\bX,\widetilde{\bbeta}_{n}))$, as well as the sample standard deviations of prediction errors in squared loss, with $p=50$.
It can be seen that larger sample size leads to much smaller mean square errors.
Moreover, when $\tau_L$ is larger, the prediction also becomes worse, and it may be due to the lower estimation efficiency.
Finally, similar to the experiments in the previous subsection, the prediction at $\tau^*= 0.991 $ is more accurate for almost all cases.
The results for the cases with $p=100$ and 150 are similar and hence omitted.

Finally, we consider the following criteria to evaluate the performance of variable selection: average number of selected active coefficients (size), percentage of active and inactive coefficients both correctly selected simultaneously ($\textrm{P}_{\textrm{AI}}$), percentage of active coefficients correctly selected ($\textrm{P}_{\textrm{A}}$), percentage of inactive coefficients correctly selected ($\textrm{P}_{\textrm{I}}$), false positive rate (FP), and false negative rate (FN).
Table \ref{tab:high2} reports the selecting results with $p=50$ and $c=10$, 30 and 50.
When $\tau_L$ is larger, both $\textrm{P}_{\textrm{AI}}$ and $\textrm{P}_{\textrm{A}}$ decrease, and it indicates the increasing of selection accuracy. In addition, performance improves when sample size gets larger. The results for $p=100$ and 150 are similar and hence omitted.

\section{Application to Childhood Malnutrition}

Childhood malnutrition is well known to be one of the most urgent problems in developing countries. The Demographic and Health Surveys (DHS) has conducted nationally representative surveys on child and maternal health, family planning and child survival, etc., and this results in many datasets for research purposes. 
The dataset for India was first analyzed by \cite{koenker2011additive}, and can be downloaded from \url{http://www.econ.uiuc.edu/~roger/research/bandaids/bandaids.html}.
It has also been studied by many researchers \citep{fenske2011identifying,belloni2019valid} for childhood malnutrition problem in India, and quantile regression with low- or high-dimensional covariates was conducted at the levels of $\tau=0.1$ and $0.05$.
The proposed model enables us to consider much lower quntiles, corresponding to more severe childhood malnutrition problem.

The child's height is chosen as the indicator for malnutrition as in \cite{belloni2019valid}. Specifically, the response  is set to $Y=-100\log(\text{child's height in centimeters})$, and we then consider high quantiles to study the childhood malnutrition problem such that it is consistent with previous sections.
Other variables include seven continuous and 13 categorical ones, and they contain both biological factors and socioeconomic factors that are possibly related to childhood malnutrition.
Examples of biological factors include the child's age, gender, duration of breastfeeding in months, the mother's age and body-mass index (BMI),
and socioeconomic factors contain the mother's employment status, religion, residence, and the availability of electricity.
All seven continuous variables are standardized to have mean zero and variance one, and two-way interactions between all variables are also included. Moreover, we concentrate on the samples from pool families.
As a result, there are $p=328$ covariates in total after removing variables with all elements being zero, and the sample size is $n=6858$.
Denote the full model size by $(328,328,328)$, which correspond to the sizes of location, scale and tail, respectively.  Furthermore, as in the simulation experiments, covariates are further rescaled to the range $[-0.5,0.5]$ for the tail index.

We aim at two high quanitles of $\tau^*=0.991$ and $0.995$, and the algorithm for high-dimensional regularized estimation in Section 3 is first applied to select the interval of $[\tau_L,\tau_U]$.
Specifically, the value of $\tau_U$ is fixed to $0.99$, and that of $\tau_L$ is selected among $\tau_L=0.9+0.01j$ with $1\leq j\leq 8$.
The value of $K$ is set to 10, and the $\tau_k$'s with $1\leq k\leq K$ are equally spaced over $[\tau_L,\tau_U]$.
For each $\tau_L$, the whole samples are randomly split into five parts with equal size, except that one part is short of two observations, and the 5-fold cross validation is used to select the tuning parameter $\lambda$.
To stabilize the process, we conduct the random splitting five times and choose the value of $\lambda$ minimizing the composite check loss over all five splittings.
The averaged value of PEs is also calculated over all five splittings, and the corresponding plot is presented in Figure \ref{fig:compare}. As a result, we choose $\tau_L=0.96$ since it corresponds to the minimum value of PEs. 

We next apply the QIR model to the whole dataset with $[\tau_L,\tau_U]=[0.96,0.99]$, and the tuning parameter $\lambda$ is scaled by$\sqrt{4/5}$ since the sample size changes from $4n/5$ to $n$. The fitted model is of size $(14,16,19)$, and we can predict the conditional quantile structure at any level $\tau^*\in(0.96,1)$.
For example, consider the variable of child's age, and we are interested in children with ages of 20, 30 and 40 months.
The duration of breastfeeding is set to be the same as child's age, since the age is always larger than the duration of breastfeeding, and we set the values of all other variables in $\bm{X}$ to be the same as the $460$th observation, which has the response value being the sample median.
Figure \ref{fig:interpretation} plots the predicted quantile curves for three different ages. It can be seen that younger children may have extremely lower heights, and we may conclude that it may be easier for younger children to be affected by malnutrition. 

Figure \ref{fig:interpretation} also draws quantile curves for mother's education, child's gender and mother's unemployment condition, and the values of variables at the $460$th observation are also used for non-focal covariates in the prediction.
For child's gender, the baby boy is usually higher than baby girls as observed in \cite{koenker2011additive}, while the difference varnishes for much larger quantiles.
In addition, the quanitle curves for mother's education are almost parallel, while those for mother's unemployment condition are crossed.
More importantly, all these new insights are at very high quantiles, and this confirms the necessity of the proposed model.

Finally, we compare the proposed QIR model with two commonly used ones in the literature: (i.) linear quantile regression at the level of $\tau^*$ with $\ell_1$ penalty in \cite{belloni2019valid}, and (ii.) extremal quantile regression in \citep{wang2012estimation} adapted to high-dimensional data.
The prediction performance at $\tau^*=0.991$ and $0.995$ is considered for the comparison, and we fix $[\tau_L,\tau_U]=[0.96,0.99]$.
For Method (ii.), we consider $K=4.5n^{1/3}$ quantile levels, equally spaced over $[0.96,0.99]$, and the linear quantile regression with $\ell_1$ penalty is conducted at each level. As in \cite{wang2012estimation}, we can estimate the extreme value index, and hence the fitted structures can be extrapolated to the level of $\tau^*$.
Note that there is no theoretical justification for Method (ii.) in the literature. 
As in simulation experiments, the tuning parameter $\lambda$ in the above three methods is selected by minimizing the composite check loss in the testing set.
We randomly split the data 100 times, and one value of PE can be obtained from each splitting.
Figure \ref{fig:compare} gives the boxplots of PEs from our model and two competing methods, and the advantages of our model can be observed at both target levels of $\tau^*=0.991$ and 0.995.

\section{Conclusions and discussions}

This paper proposes a reliable method for the inference at extreme quantiles with both low- and high-dimensional covariates.
The main idea is first to conduct a composite quantile regression at fixed quantile levels, and we then can extrapolate the estimated results to extreme quantiles by assuming a parametric structure at tails. 
The Tukey lambda structure can be used due to its flexibility and the explicit form of its quantile functions, and the success of the proposed methodology has been demonstrated by extensive numeral studies. 

This paper can be extended in the following two directions.
On one hand, in the proposed model, a parametric structure is assumed over the interval of $[\tau_0,1]$. Although the criterion of PE is suggested in Section 3 to balance the estimation efficiency and model misspecification, it should be interesting to provide a statistical tool for the goodness-of-fit. \cite{Dong_Feng_Li2019} introduced a goodness-of-fit test for parametric quantile regression at a fixed quantile level, and it can be used for our problem by extending the test statistic from a fixed level to the interval of $[\tau_0,1]$. We leave it for the future research.
On the other hand, the idea in this paper is general and can be applied to many other scenarios. 
For example, for conditional heteroscedastic time series models, it is usually difficult to conduct the quantile estimation at both median and extreme quantiles. The difficulty at extreme quantiles is due to the sparse data at tails, while that at median is due to the tiny values of fitted parameters \citep{Zhu_Zheng_Li2018,Zhu_Li2021}.
Our idea certainly can be used to solve this problem to some extent.


\renewcommand{\thesection}{A}
\renewcommand{\theequation}{\thesection.\arabic{equation}}
\renewcommand{\thelemma}{\thesection.\arabic{lemma}}
\setcounter{equation}{0}
\setcounter{lemma}{0}

\section*{Appendix: Technical details}

\subsection{Proofs for low-dimensional CQR Estimation}

This subsection gives technical proofs of Lemma \ref{lemma:tukey} and Theorems \ref{thm-identification}-\ref{theorem2} in section 2.2.
Two auxiliary lemmas are also presented at the end of this subsection: Lemma \ref{lemma:A1} is used for the proof of Lemma \ref{lemma:tukey}, and Lemma \ref{lem-thm3} is for that of Theorem \ref{theorem2}.
The proofs of these two auxiliary lemmas are given in a supplementary file.

\begin{proof}[{Proof of Lemma \ref{lemma:tukey}}]
	The Tukey lambda distribution has the form of
	\[
	Q(\tau,\btheta)=\theta_{1}+{\theta_{2}}\left( \frac{\tau^{\theta_{3}}-(1-\tau)^{\theta_{3}}}{\theta_{3}}\right),
	\]
	where $\theta_1\in\mathbb{R}$, $\theta_2>0$, $\theta_3\ne0$. We prove Lemma \ref{lemma:tukey} for $\theta_{3}<1$.
	Consider four arbitrary quantile levels $0<\tau_1<\tau_2<\tau_3<\tau_4<1$, and two arbitrary index vectors $\widetilde{\bm\theta}=(\widetilde{\theta}_1,...,\widetilde{\theta}_4)^{\prime}$ and $\bm\theta=(\theta_1,..,\theta_4)^{\prime}$ such that
	\begin{equation}\label{thm1-proof-eq1}
	Q(\tau_j,\widetilde{\bm\theta})=Q(\tau_j,\bm\theta) \text{ for all }1\leq j\leq 4.
	\end{equation}
	We show that $\widetilde{\bm\theta}=\bm\theta$ in the following.
	
	The first step is to prove $\widetilde{\theta}_3=\theta_{3}$ using the proof by contradiction.
	Suppose that $\widetilde{\theta}_3\ne\theta_{3}$ and, without loss of generality, we assume that $\theta_{3}<\widetilde{\theta}_3<0$.
	Denote
	$f_j(\theta_3)=\tau_{j}^{\theta_3}-(1-\tau_j)^{\theta_3}$ for $1\leq j\leq 4$.
	From \eqref{thm1-proof-eq1}, we have
	\begin{equation}\label{Eq:2q}
	\frac{f_{1}(\widetilde{\theta}_3)-f_{2}(\widetilde{\theta}_3)}{f_{3}(\widetilde{\theta}_3)-f_{2}(\widetilde{\theta}_3)} -\frac{f_{1}(\theta_3)-f_{2}(\theta_3)}{f_{3}(\theta_3)-f_{2}(\theta_3)}=0,
	\end{equation}
	and
	\begin{equation}\label{Eq:3q}
	\frac{f_{4}(\widetilde{\theta}_3)-f_{2}(\widetilde{\theta}_3)}{f_{3}(\widetilde{\theta}_3)-f_{2}(\widetilde{\theta}_3)} -\frac{f_{4}(\theta_3)-f_{2}(\theta_3)}{f_{3}(\theta_3)-f_{2}(\theta_3)}=0.
	\end{equation}
	Let us fix $\theta_3$, $\widetilde{\theta}_3$, $\tau_2$ and $\tau_3$. As a result, $\kappa_1=f_{2}(\widetilde{\theta}_3)$, $\kappa_2=f_{2}(\theta_3)$, $\kappa_3=f_{3}(\theta_3)-f_{2}(\theta_3)$ and $\kappa_4=f_{3}(\widetilde{\theta}_3)-f_{2}(\widetilde{\theta}_3)$ are all fixed values.
	Denote
	\begin{align}
	F(\tau)=\kappa_3\left\{\tau^{\widetilde{\theta}_3}-(1-\tau)^{\widetilde{\theta}_3} -\kappa_1\right\} -\kappa_4\left\{\tau^{\theta_3}-(1-\tau)^{\theta_3} -\kappa_2\right\},
	\end{align}
	\[
	\dot{F}(\tau)=\kappa_3\widetilde{\theta}_3\{\tau^{\widetilde{\theta}_3-1}+(1-\tau)^{\widetilde{\theta}_3-1}\}-\kappa_4\theta_{3}\{\tau^{{\theta}_3-1}+(1-\tau)^{{\theta}_3-1}\},
	\]
	and
	\[
	G(\tau)=\frac{\tau^{\widetilde{\theta}_3-1}+(1-\tau)^{\widetilde{\theta}_3-1}}{\tau^{{\theta}_3-1}+(1-\tau)^{{\theta}_3-1}},
	\]
	where $\dot{F}(\cdot)$ is the derivative function of $F(\cdot)$, and $\dot{F}(\tau)=0$ if and only if $G(\tau)=\kappa_4\theta_{3}/(\kappa_3\widetilde{\theta}_3)$.
	Note that equations \eqref{Eq:2q} and \eqref{Eq:3q} correspond to $F(\tau_1)=0$ and $F(\tau_4)=0$, respectively. Moreover, it can be verified that $F(\tau_2)=0$ and $F(\tau_3)=0$, i.e. the equation $F(\tau)=0$ has at least four different solutions.
	As a result, the equation $\dot{F}(\tau)=0$ or $G(\tau)=\kappa_4\theta_{3}/(\kappa_3\widetilde{\theta}_3)$ has at least three different solutions.
	While it is implied by Lemma \ref{lemma:A1} that the equation $G(\tau)=\kappa_4\theta_{3}/(\kappa_3\widetilde{\theta}_3)$ has at most two different solutions. Due to the contradiction, we prove that $\widetilde{\theta}_3=\theta_{3}$, and it is readily to further verify that $(\widetilde{\theta}_1,\widetilde{\theta}_2)=(\theta_{1},\theta_2)$. We hence accomplish the proof of Lemma \ref{lemma:tukey}(i). The result of Lemma \ref{lemma:tukey}(ii) can be proved similarly.
\end{proof}

\begin{proof}[{Proof of Theorem \ref{thm-identification}}]
	We first prove the uniqueness of $\beta_0$.
	Denote $\bar{L}_{k}(\bbeta)=E[\rho_{\tau_{k}}\{Y-Q(\tau_{k},\btheta(\bX,\bbeta))\}]$. From model (\ref{Model2}), $\bbeta_{0}$ is the minimizer not only of $\bar{L}(\bbeta)=\sum_{k=1}^{K}\bar{L}_k(\bbeta)$, but also of $\bar{L}_k(\bbeta)$ for all $1\le k\le K$. Suppose that $\bbeta_{0}^*$ is another minimizer of $\bar{L}(\bbeta)$, and then it is also the minimizer of $\bar{L}_k(\bbeta)$ for all $1\le k\le K$.
	
	Note that, for $u\neq 0$,
	\begin{align}
		\begin{split}\label{knight}
			\rho_{\tau}(u-v)-\rho_{\tau}(u)&=-v\psi_{\tau}(u)+(u-v)[I(0>u>v)-I(0<u<v)]\\
			&=-v\psi_{\tau}(u)+\int_{0}^{v}[I(u\leq s)-I(u\leq 0)]ds,
		\end{split}
	\end{align}
	where $\psi_{\tau}(u)=\tau-I(u<0)$; see \cite{Knight1998}.
	Let $U^{(k)}=Y-Q(\tau_{k},\btheta(\bX,\bbeta_0))$ and $V^{(k)}=Q(\tau_{k},\btheta(\bX,\bbeta_0^*))-Q(\tau_{k},\btheta(\bX,\bbeta_0))$. It holds that, for each $1\leq k\leq K$,
	\[
	0=\bar{L}_k(\bbeta_0^*)-\bar{L}_k(\bbeta_{0})=E\{(U^{(k)}-V^{(k)})[I(0>U^{(k)}>V^{(k)})-I(0<U^{(k)}<V^{(k)})]\},
	\]
	which implies that $V^{(k)}=Q(\tau_{k},\btheta(\bX,\bbeta_0^*))-Q(\tau_{k},\btheta(\bX,\bbeta_0))=0$ with probability one. This, together with the conditions in Lemma \ref{lemma:tukey} and the monotonic link functions, leads to the fact that $\bX'{\bbeta}_{j0}^*=\bX'{\bbeta}_{j0}$ for all $1\leq j\leq d$. We then have $\bm\beta_0^*=\bm\beta_0$ since $E(\bm X\bm X^{\prime})$ is positive definite. This accomplishes the proof of the uniqueness of $\bm\beta_0$.
\end{proof}

\begin{proof}[{Proof of Theroem \ref{theorem1}}]
	
Note that $L_n(\bbeta)=n^{-1}\sum_{k=1}^{K}\sum_{i=1}^{n}\rho_{\tau_{k}}\left\{Y_{i}-Q(\tau_{k},\btheta(\bX_{i},\bbeta))\right\}$ and $ \bar{L}(\bbeta)=E[\sum_{k=1}^{K}\rho_{\tau_{k}}\{Y-Q(\tau_{k},\btheta(\bX,\bbeta))\}]$. From Knight's identity at \eqref{knight},
\begin{align*}
\bar{L}(\bbeta)-\bar{L}(\bbeta_{0})&\le\sum_{k=1}^{K}E|Q(\tau_{k},\btheta(\bX,\bbeta))-Q(\tau_{k},\btheta(\bX,\bbeta_{0}))|\notag\\
&\le\sum_{k=1}^{K} E\sup_{\bbeta\in\Theta}\|\frac{\partial Q(\tau_{k},\btheta(\bX,\bbeta))}{\partial \bbeta}\|\|\bbeta-\bbeta_{0}\|,
\end{align*}	
which, together with the condition of $E\max_{1\leq k\leq K}\sup_{\bbeta\in\Theta}\|{\partial Q(\tau_{k},\btheta(\bX,\bbeta))}/{\partial \bbeta}\|<\infty$ and Corollary 2 of \cite{andrews1987consistency}, implies that
\begin{equation}\label{eq:thm2-1}
\sup_{\bbeta\in\Theta}|L_n(\bbeta)-L_n(\bbeta_{0})-[\bar{L}(\bbeta)-\bar{L}(\bbeta_{0})]|=o_p(1).
\end{equation}

Note that $\bar{L}(\bbeta)$ is a continuous function with respect to $\bbeta$ and, from Theorem \ref{thm-identification}, $\bbeta_0$ is the unique minimizer of $\bar{L}(\bbeta)$. As a result, for each $\delta>0$,
\begin{equation}
\epsilon=\inf_{\bbeta\in B_{\delta}^c(\bbeta_{0})}\bar{L}(\bbeta)-\bar{L}(\bbeta_0)>0,
\end{equation}
where $B_{\delta}(\bbeta_{0})=\{\bbeta: \|\bbeta-\bbeta_{0}\|\leq \delta\}$ and $B_{\delta}^c(\bbeta_{0})$ is its complement set, and hence
\begin{equation}\label{eq:thm2-2}
\left\{\inf_{\bbeta\in B_{\delta}^c(\bbeta_{0})}L_n(\bbeta)\le L_n(\bbeta_{0})\right\} \subseteq\left\{\sup_{\bbeta\in B_{\delta}^c(\bbeta_{0})}|L_n(\bbeta)-L_n(\bbeta_{0})-[\bar{L}(\bbeta)-\bar{L}(\bbeta_{0})]\}|\geq \epsilon \right\}.
\end{equation}

Note that
\[
1=P\left\{L_n(\widehat{\bbeta}_n)\le L_n(\bbeta_{0})\right\}\leq P\left\{\widehat{\bbeta}_n\in B_{\delta}(\bbeta_{0})\right\}+P\left\{\inf_{\bbeta\in B_{\delta}^c(\bbeta_{0})}L_n(\bbeta)\le L_n(\bbeta_{0})\right\}
\]
which together with \eqref{eq:thm2-1} and \eqref{eq:thm2-2}, implies that
\[
P\left\{\|\bbeta-\bbeta_{0}\|\leq \delta \right\}\geq 1 -P\left\{\inf_{\bbeta\in B_{\delta}^c(\bbeta_{0})}L_n(\bbeta)\le L_n(\bbeta_{0})\right\} \rightarrow 1,
\]
as $n\rightarrow \infty$. This accomplishes the proof of consistency.

\end{proof}

\begin{proof}[{Proof of Theorem \ref{theorem2}}]
For simplicity, we denote $Q(\tau_{k},\btheta(\bX_{i},\bbeta))$ by $Q_{ik}(\bbeta)$ in the whole proof of this theorem. Let $S_n(\bbeta)=L_n(\bbeta)-L_n(\bbeta_0)$ and, from Knight's identity at \eqref{knight},
\begin{equation}
\begin{split}
S_n(\bbeta)=&\sum_{k=1}^{K}\sum_{i=1}^{n}\left[\rho_{\tau_{k}}\left\{Y_{i}-Q_{ik}(\bbeta)\right\}-\rho_{\tau_{k}}\left\{Y_{i}-Q_{ik}(\bbeta_{0})\right\}\right]\\
=&\sum_{k=1}^{K}\sum_{i=1}^{n}\Big[\left\{Q_{ik}(\bbeta)-Q_{ik}(\bbeta_{0})\right\}\left\{I(e_{ik}<0)-\tau_{k}\right\} \\
&\hspace{17mm}+\int_{0}^{Q_{ik}(\bbeta)-Q_{ik}(\bbeta_{0})}\{I(e_{ik}\le s)-I(e_{ik}\le 0)\}ds\Big].
\end{split}
\end{equation}
Note that, by Taylor expansion, $Q_{ik}(\bbeta)-Q_{ik}(\bbeta_{0})= (\bbeta-\bbeta_0)^{\prime}{\partial Q_{ik}(\bbeta_{0})}/{\partial\bbeta} +(\bbeta-\bbeta_0)^{\prime}({\partial^2 Q_{ik}(\bbeta^{*})}/{\partial\bbeta\partial\bbeta^{\prime}} )(\bbeta-\bbeta_0)$, where $\bbeta^{*}$ is a vector between $\bbeta_{0}$ and $\bbeta$.
Let $u=\bbeta-\bbeta_{0}$,
\begin{equation*}
q_{1ik}(u)=u'\frac{\partial Q_{ik}(\bbeta_{0})}{\partial\bbeta}\hspace{3mm}\text{and}\hspace{3mm} q_{2ik}(u)=u'\frac{\partial^{2}Q_{ik}(\bbeta^{*})}{\partial\bbeta\partial\bbeta'}u.
\end{equation*}
We then decompose $S_n(\bbeta)$ into
\begin{align}
S_n(\bbeta)&=\sum_{k=1}^{K}\sum_{i=1}^{n}\Big[\left\{q_{1ik}(u)+q_{2ik}(u)\right\}\{I(e_{ik}<0)-\tau_{k}\} \\
&\hspace{17mm}+\int_{0}^{q_{1ik}(u)+q_{2ik}(u)}\{I(e_{ik}\le s)-I(e_{ik}\le 0)\}ds\Big]\\
&=-u'T_{n}+\Pi_{1n}(u)+\Pi_{2n}(u)+\Pi_{3n}(u),
\end{align}
where
\begin{eqnarray}
&&T_{n}=\sum_{k=1}^{K}\sum_{i=1}^{n}\frac{\partial Q_{ik}(\bbeta_{0})}{\partial\bbeta}\{\tau_{k}-I(e_{ik}\le 0)\},\notag\\
&&\xi_{ik}(u)=\int_{0}^{q_{1ik}(u)}\{I(e_{ik}\le s)-I(e_{ik}\le 0)\}ds,\notag\\
&&\Pi_{1n}(u)=\sum_{k=1}^{K}\sum_{i=1}^{n}\left[\xi_{ik}(u)-E\left\{\xi_{ik}(u)|\bX_{i}\right\}\right],\notag\\
&&\Pi_{2n}(u)=\sum_{k=1}^{K}\sum_{i=1}^{n}E\left\{\xi_{ik}(u)|\bX_{i}\right\},\notag
\end{eqnarray}
and
\[
\Pi_{3n}(u)=\sum_{k=1}^{K}\sum_{i=1}^{n}\Big[q_{2ik}(u)\left\{I(e_{ik}<0)-\tau_{k}\right\}+\int_{q_{1ik}(u)}^{q_{1ik}(u)+q_{2ik}(u)}\left\{I(e_{ik}\le s)-I(e_{ik}\le 0)\right\}ds\Big].
\]

First, by the central limit theorem, we can show that
\begin{equation}\label{eq:thm3-2}
\frac{1}{\sqrt{n}}T_{n}=\frac{1}{\sqrt{n}}\sum_{k=1}^{K}\sum_{i=1}^{n}\frac{\partial Q_{ik}(\bbeta_{0})}{\partial\bbeta}\{\tau_{k}-I(e_{ik}\le 0)\} \rightarrow N(0, \Omega_0)
\end{equation}
in distribution as $n\rightarrow\infty$.
Note that, from Theorem \ref{theorem1}, $\widehat{u}_n=\widehat{\bbeta}_n-\bbeta_{0}=o_{p}(1)$ and then, by applying Lemma \ref{lem-thm3},
\begin{equation}
\label{expansion}
\begin{split}
S_n(\widehat{\bbeta}_n)&=-\widehat{u}_n^{\prime}T_{n}+\Pi_{1n}(\widehat{u}_n)+\Pi_{2n}(\widehat{u}_n)+\Pi_{3n}(\widehat{u}_n)\\
&=-\widehat{u}_n^{\prime}T_{n}+\frac{1}{2}(\sqrt{n}\widehat{u}_n)^\prime\Omega_{1}(\sqrt{n}\widehat{u}_n)+o_{p}(\sqrt{n}\|\widehat{u}_n\|+n\|\widehat{u}_n\|^{2})\\
&\ge-\|\sqrt{n}\widehat{u}_n\|\left\{\|\frac{1}{\sqrt{n}}T_{n}\|+o_{p}(1)\right\}+n\|\widehat{u}_n\|^{2}\left\{\frac{\lambda_{\min}(\Omega_{1})}{2}+o_{p}(1)\right\},
\end{split}
\end{equation}
where $\lambda_{\min}(\Omega_{1})$ is the minimum eigenvalue of $\Omega_{1}$.
This, together with the fact that $S_n(\widehat{\bbeta}_n)=L_n(\widehat{\bbeta}_n)-L_n(\bbeta_0)\leq 0$, implies that
\begin{equation}\label{eq:thm3-1}
\sqrt{n}\|\widehat{u}_n\|\leq \left\{\frac{\lambda_{\min}(\Omega_{1})}{2}+o_{p}(1)\right\}^{-1}\left\{\|\frac{1}{\sqrt{n}}T_{n}\|+o_{p}(1)\right\} =O_{p}(1).
\end{equation}

Denote $u_n^{\star}=n^{-1}\Omega_{1}^{-1}T_{n}$ and, from \eqref{expansion} and \eqref{eq:thm3-1},
\begin{equation}
S_n(\widehat{\bbeta}_n)=\frac{1}{2}(\sqrt{n}\widehat{u}_n)^{\prime}\Omega_{1}(\sqrt{n}\widehat{u}_n)-(\sqrt{n}\widehat{u}_n)^{\prime}\Omega_{1}(\sqrt{n}u_n^{\star})+o_{p}(1).
\end{equation}
Moreover, since $\sqrt{n}u_n^{\star}=O_p(1)$, equation \eqref{expansion} still holds when $\widehat{u}_n$ is replaced by $u_n^{\star}$, and then
\begin{equation}
S(\bbeta_0+u_n^{\star})=-\frac{1}{2}(\sqrt{n}u_n^{\star})^{\prime}\Omega_{1}(\sqrt{n}u_n^{\star})+o_{p}(1),
\end{equation}
which leads to
\begin{eqnarray}
0\geq S(\widehat{\bbeta}_n)-S(\bbeta_0+u_n^{\star})&=&\frac{1}{2}(\sqrt{n}\widehat{u}_n-\sqrt{n}u_n^{\star})'\Omega_{1}(\sqrt{n}\widehat{u}_n-\sqrt{n}u_n^{\star})+o_{p}(1)\notag\\
&\geq&\frac{\lambda_{\min}(\Omega_{1})}{2}\|\sqrt{n}\widehat{u}_n-\sqrt{n}u_n^{\star}\|^{2}+o_{p}(1).\notag
\end{eqnarray}
As a result, from \eqref{eq:thm3-2},
\[
\sqrt{n}\widehat{u}_n=\sqrt{n}u_n^{\star}+o_p(1)=\Omega_{1}^{-1} n^{-1/2}T_{n} +o_p(1) \rightarrow N(\bm{0},\Omega_{1}^{-1}\Omega_{0}\Omega_{1}^{-1})
\]
in distribution as $n\rightarrow\infty$.
\end{proof}

\begin{lemma}\label{lemma:A1}
	Consider the function of $G(\tau)$ defined in the proof of Lemma \ref{lemma:tukey} with $\theta_{3}<\widetilde{\theta}_3<0$. It holds that, (1) for $\tau>0.5$, $G(\tau)$ is strictly decreasing,  (2) for $\tau<0.5$, $G(\tau)$ is strictly increasing.
\end{lemma}

\begin{lemma}\label{lem-thm3}
	Suppose that the conditions of Theorem \ref{theorem2} hold. For any sequence of random variables $\{u_n\}$ with $u_{n}=o_{p}(1)$, it holds that
	\begin{itemize}
		\item [(a)] $\Pi_{1n}(u_{n})=o_{p}(\sqrt{n}\|u_{n}\|+n\|u_{n}\|^{2})$,
		\item [(b)] $\Pi_{2n}(u_{n})=(1/2)(\sqrt{n}u_{n})'\Omega_{1}(\sqrt{n}u_{n})+o_{p}(n\|u_{n}\|^{2})$, and
		\item [(c)] $\Pi_{3n}(u_{n})=o_{p}(n\|u_{n}\|^{2})$,
	\end{itemize}
where $\Pi_{1n}(u)$, $\Pi_{2n}(u)$ and $\Pi_{3n}(u)$ are defined in the proof of Theorem \ref{theorem2}.
\end{lemma}




\subsection{Proofs for high-dimensional regularized estimation}

This subsection first conducts deterministic analysis at Lemma \ref{deterministic-analysis}, and then  stochastic analysis at Lemmas \ref{infty} and \ref{RSC}. The proof of Theorem \ref{high-dim} follows from the deterministic analysis in Lemma \ref{deterministic-analysis} and stochastic analysis in Lemmas \ref{infty}, and \ref{RSC}.

We first treat the observed data, $\{(Y_i,\bm X_i^{\prime})^{\prime},i=1,...,n\}$, to be deterministic.
Consider the loss function $L_n(\bbeta)=\sum_{i=1}^{n}\sum_{k=1}^{K}\rho_{\tau_{k}}\{Y_{i}-Q(\tau_{k},\btheta(\bm X_i,\bbeta))\}$, and its first-order Taylor-series error
\[
\mathcal{E}_n(\Delta) = n^{-1}L_n(\bbeta_0+\Delta) -n^{-1}L_n(\bbeta_0) -\langle n^{-1}\nabla L_n(\bbeta_0),\Delta\rangle,
\]
where $\Delta\in\mathbb{R}^{dp}$, $\langle \cdot,\cdot\rangle$ is the inner product, $e_{ik}=Y_{i}-Q(\tau_{k},\btheta(\bm X_i,\bbeta_0))$, $\psi_{\tau}(u)=\tau-I(u<0)$, and $\nabla L_n(\bbeta)=\sum_{i=1}^{n}\sum_{k=1}^{K} \psi_{\tau}(e_{ik}) {\partial Q(\tau_{k},\btheta(\bX,\bbeta))}/{\partial\bbeta}$ is a subgradient of $L_n(\bbeta)$.

\begin{definition}\label{def:LRSC}
	Loss function $L_n(\cdot)$ satisfies the local restricted strong convexity (LRSC) condition if
	\[
	\mathcal{E}_n(\Delta) \geq \alpha \|\Delta\|_2^2- \eta \sqrt{\frac{\log p}{n}} \|\Delta\|_1 \hspace{2mm}\text{for all $\Delta$ such that} \hspace{2mm} 0<r \leq \|\Delta\|_2\leq R,
	\]
	where $\alpha,\eta>0$, and $\|\cdot\|_1$ and $\|\cdot\|_2$ are $\ell_1$ and $\ell_2$ norms, respectively.
\end{definition}

The above LRSC condition has a larger tolerance term compared with that in \cite{loh2015regularized}, which has a form of $(\log p/n)\|\Delta\|_{1}^{2}$. Similar tolerance term can also be found in \cite{fan2019generalized} for high-dimensional generalized trace regression.
It is ready to establish an upper bound for estimation errors when the penalty $\lambda$ is appropriately selected.

\begin{lemma}\label{deterministic-analysis}
	Suppose that the regularizer $p_{\lambda}(\cdot)$ satisfies Assumption \ref{regular}, and loss function $L_n(\cdot)$ satisfies the LRSC condition with $\alpha>\mu/4$ and $r=\frac{12\eta\sqrt{s}}{(4\alpha-\mu)L}\sqrt{\frac{\log p}{n}}$. If the tuning parameter $\lambda$ satisfies that
	\begin{equation}
	\lambda\ge\frac{4}{L}\max\left\{\|n^{-1}\nabla L(\bbeta_{0})\|_{\infty},\eta\sqrt{\frac{\log p}{n}}\right\},
	\end{equation}
	then the minimizer $\widetilde{\bbeta}_n$ over the set of $\Theta=\mathcal{B}_{R}(\bbeta_{0})$ satisfies the error bounds
	\begin{align}
	\|\widetilde{\bbeta}_n-\bbeta_{0}\|_{2}\le\frac{6\sqrt{s}\lambda L}{4\alpha-\mu}~~~~~\mbox{and}~~~~~\|\widetilde{\bbeta}_n-\bbeta_{0}\|_{1}\le\frac{24s\lambda L}{4\alpha-\mu}.
	\end{align}
\end{lemma}
\begin{proof}[{Proof.}]
	Denote $\widetilde{\Delta}_n=\widetilde{\bbeta}_n-\bbeta_{0}$, and it holds that $\|\widetilde{\Delta}_n\|_2\leq R$.
	Note that this lemma naturally holds if $\|\widetilde{\Delta}_n\|_{2}\le (3\sqrt{s}\lambda)/(4\alpha-\mu)$.
	As a result, we only need to consider the case with $(3\sqrt{s}\lambda)/(4\alpha-\mu)\le \|\widetilde{\Delta}_n\|_{2}\le R$, and it can be verified that
	\begin{equation}
	r=\frac{12\eta\sqrt{s}}{(4\alpha-\mu)L}\sqrt{\frac{\log p}{n}}\le\|\widetilde{\Delta}_n\|_{2}\le R.
	\end{equation}
	
	Note that $n^{-1}L_n(\widetilde{\bbeta}_n)+p_{\lambda}(\widetilde{\bbeta}_n)\leq n^{-1}L_n(\bbeta_{0})+p_{\lambda}(\bbeta_{0})$, and then
	$\mathcal{E}_n(\widetilde{\Delta}_n) \leq p_{\lambda}(\bbeta_{0})-p_{\lambda}(\widetilde{\bbeta}_n) -\langle n^{-1}\nabla L_n(\bbeta_0),\Delta\rangle $. This, together with the LRSC condition and Holder's inequality, implies that
	\begin{align}
	\alpha\|\widetilde{\Delta}_n\|_{2}^{2}\le& p_{\lambda}(\bbeta_{0})-p_{\lambda}(\widetilde{\bbeta}_n)+\left(\eta\sqrt{\frac{\log p}{n}} + \|n^{-1}\nabla L(\bbeta_{0})\|_{\infty}\right)\|\widetilde{\Delta}_n\|_{1}\\
	\leq & p_{\lambda}(\bbeta_{0})-p_{\lambda}(\widetilde{\bbeta}_n)+\frac{\lambda L}{2}\|\widetilde{\Delta}_n\|_{1}\\
	\le&p_{\lambda}(\bbeta_{0})-p_{\lambda}(\widetilde{\bbeta}_n)+\frac{1}{2}\left\{p_{\lambda}(\widetilde{\bbeta}_n-\bbeta_{0})+\frac{\mu}{2}\|\widetilde{\Delta}_n\|_{2}^{2}\right\}\\
	\le&p_{\lambda}(\bbeta_{0})-p_{\lambda}(\widetilde{\bbeta}_n)+\frac{1}{2}\left\{p_{\lambda}(\widetilde{\bbeta}_n)+p_{\lambda}(\bbeta_{0})+\frac{\mu}{2}\|\widetilde{\Delta}_n\|_{2}^{2}\right\},
	\end{align}
	where the last inequality follows from the subadditivity of $p_{\lambda}(\cdot)$, while the penultimate inequality is by Assumption \ref{regular}; see also Lemma 4 in \cite{loh2015regularized}.
	As a result,
	\begin{align}\label{3p-p-1}
	0< \left(\alpha-\frac{\mu}{4}\right)\|\widetilde{\Delta}_n\|_{2}^{2}\le \frac{3}{2}p_{\lambda}(\bbeta_{0})-\frac{1}{2}p_{\lambda}(\widetilde{\bbeta}_n).
	\end{align}
	Moreover, from Lemma 5 in \cite{loh2015regularized}, it holds that
	\begin{align}\label{3p-p-2}
	0\le 3p_{\lambda}(\bbeta_{0})-p(\widetilde{\bbeta}_n)\le \lambda L(3\|(\widetilde{\Delta}_n)_{A}\|_{1}-\|(\widetilde{\Delta}_n)_{A^{c}}\|_{1}),
	\end{align}
	where $A$ is the index set of the $s$ largest elements of $\widetilde{\bbeta}_n$ in magnitude.
	Combining \eqref{3p-p-1} and \eqref{3p-p-2}, we have
	\begin{align}
	\left(\alpha-\frac{\mu}{4}\right)\|\widetilde{\Delta}_n\|_{2}^{2}\le \frac{3\lambda L}{2}\|(\widetilde{\Delta}_n)_{A}\|_{1}\le \frac{3\lambda L\sqrt{s}}{2}\|\widetilde{\Delta}_n\|_{2}.
	\end{align}
	As a result,
	\begin{align}
	\|\widetilde{\Delta}_n\|_{2}\le\frac{6\sqrt{s}\lambda L}{4\alpha-\mu}.
	\end{align}
	
	It is also implied by \eqref{3p-p-2} that $\|(\widetilde{\Delta}_n)_{A^{c}}\|_{1}
	\leq 3\|(\widetilde{\Delta}_n)_{A}\|_{1}$, which leads to
	\begin{equation}
	\|\widetilde{\Delta}_n\|_{1}\le \|(\widetilde{\Delta}_n)_{A}\|_{1}+\|(\widetilde{\Delta}_n)_{A^{c}}\|_{1}\le 4\|(\widetilde{\Delta}_n)_{A}\|_{1}\le 4\sqrt{s}\|\widetilde{\Delta}_n\|_{2}.
	\end{equation}
	This accomplishes the proof of this lemma.
\end{proof}

We next conduct the stochastic analysis to verify that the ``good" event and LRSC condition hold with high probability in Lemmas \ref{infty} and \ref{RSC}, respectively.

\begin{lemma}\label{infty}
	If Assumption \ref{Cond-qf} holds, then
	\begin{equation}\label{def-Cs}
	\|n^{-1}\nabla L(\bbeta_{0})\|_{\infty}\le C_{S}\sqrt{\frac{\log p}{n}}
	\end{equation}
	with probability at least $1-c_{1}p^{-c_{2}}$ for some positive constants $c_{1}$, $c_{2}$ and $C_{S}$.
\end{lemma}

\begin{proof}[{Proof.}]
	Note that
	\begin{equation}
	n^{-1}\nabla L(\bbeta_{0})=\frac{1}{n}\sum_{i=1}^{n}\sum_{k=1}^{K} \psi_{\tau_k}(e_{ik}) \frac{\partial Q(\tau_{k},\btheta(\bm X_{i},\bbeta_{0}))}{\partial\bm{\gamma}}\otimes\bm X_{i}.
	\end{equation}
	where $e_{ik}=Y_{i}-Q(\tau_{k},\btheta(\bm X_i,\bbeta_0))$, $\psi_{\tau}(u)=\tau-I(u<0)$, and $\bm X_i=(X_{1i},...,X_{pi})^\prime$.
	For $1\leq j\leq d$, denote $\xi_j(\bm X_i)={\partial Q(\tau_{k},\btheta(\bm X_{i},\bbeta_{0}))}/{\partial{\gamma_j}}$ and, from Assumption \ref{Cond-qf}, it holds that $|\xi_j(\bm X_i)|\leq L_Q$.
	
	It can be verified that, conditional on $\bm X_i$, $\psi_{\tau_k}(e_{ik})$ is sub-Gaussian with the parameter of 0.5, and hence, for any $\delta>0$,
	\[
	E\exp[\delta n^{-1}\psi_{\tau_k}(e_{ik})\xi_j(\bm X_i)X_{li}] \leq E\exp\{[\delta n^{-1}\xi_j(\bm X_i)X_{li}]^2/8\} \leq \exp\{[\delta n^{-1}L_QC_X]^2/8\}.
	\]
	As a result, for each $t>0$, $1\leq j\leq d$ and $1\leq l\leq p$,
	\begin{align}
	P&\left( \frac{1}{n}\sum_{i=1}^{n}\sum_{k=1}^{K} \psi_{\tau_k}(e_{ik}) \xi_j(\bm X_{i}) X_{li} >t\right) \\ &\leq \inf_{\delta>0} \exp(-\delta t) \prod_{i=1}^{n}\prod_{k=1}^{K}E\exp[\delta n^{-1}\psi_{\tau_k}(e_{ik})\xi_j(\bm X_i)X_{li}]\\
	&\leq \inf_{\delta>0}  \exp\left( \frac{n^{-1}K(L_QC_X)^2}{8} \delta^2-\delta t\right) \leq \exp\left( \frac{-2nt^2}{K(L_QC_X)^2}\right),
	\end{align}
	which implies that
	\begin{equation*}
	P\left(\|n^{-1}\nabla L(\bbeta_{0})\|_{\infty}\geq t \right) \leq \exp\left( \frac{-2nt^2}{K(L_QC_X)^2} +\log(2dp)\right).
	\end{equation*}
	By letting $t=C_S\sqrt{\log p/n}$ with $C_S>\sqrt{0.5K}L_QC_X$, we accomplish the proof with $c_1=2d$ and $c_2=2C_S^2/[K(L_QC_X)^2]-1>0$.
	
\end{proof}

\begin{lemma}\label{RSC}
	Suppose that Assumptions \ref{regular}-\ref{Cbeta} hold. Given the sample size $n\ge c'\log p$ for some $c'>0$, it holds that
	\[
	\mathcal{E}_n(\Delta)
	\ge \alpha\|\Delta\|_{2}^{2}-\eta\sqrt{\frac{\log p}{n}}\|\Delta\|_{1} \hspace{2mm}\text{for all $r\leq \|\Delta\|_{2}\leq R$}
	\]
	with probability at least $1-c_{1}p^{-c_{2}}-K\log(\sqrt{dp}r/r_{l})p^{-c^{2}}$ for any $c>1$, where  $\alpha=0.5f_{\min}\lambda_{\min}^{0}$, $\eta=KC_{E}d2^{d+1}+2KL_{Q}C_{X}c+C_{S}$.
\end{lemma}

\begin{proof}[{Proof.}]
	
	We first show the strong convexity of $\bar{L}(\bbeta)=E[n^{-1}L_n(\bbeta)] =E[\sum_{k=1}^{K}\rho_{\tau_{k}}\{Y-Q(\tau_{k},\btheta(\bX,\bbeta))\}]$.
	Let $Q^*(\bbeta)=(Q(\tau_{1},\btheta(\bm X,\bbeta)),\ldots,Q(\tau_{K},\btheta(\bm X,\bbeta)))^{\prime}$ and, by Taylor expansion,
	\[
	E\|Q^*(\bbeta)-Q^*(\bbeta_0)\|_2^2 =E\left\|\frac{\partial Q^*(\bbeta^*)}{\partial \bbeta}\Delta\right\|_2^2 =\Delta^{\prime}\Omega_{2}(\bbeta^*) \Delta,
	\]
	where $\Delta=\bbeta-\bbeta_0$, and $\bbeta^*$ is between $\bbeta_0$ and $\bbeta$.
	Note that ${\partial \bar{L}(\bbeta_0) }/{\partial \bbeta}=0$ and hence,
	by Knight's identity at \eqref{knight} and Assumption \ref{Cbeta}, it can be verified that
	\begin{align}
	\begin{split}\label{lemma3-eq1}
	&\bar{\mathcal{E}}(\Delta) :=\bar{L}(\bbeta)-\bar{L}(\bbeta_0)-\langle \Delta,{\partial \bar{L}(\bbeta_0) }/{\partial \bbeta}\rangle=\bar{L}(\bbeta)-\bar{L}(\bbeta_{0})\\
	=&E\left(\sum_{k=1}^{K}\int_{0}^{Q(\tau_{k},\btheta(\bm X,\bbeta))-Q(\tau_{k},\btheta(\bm X,\bbeta_{0}))}\{F_{Y|\bm X}(Q(\tau_{k},\btheta(\bm X,\bbeta_{0}))+s)-F_{Y|\bm X}(Q(\tau_{k},\btheta(\bm X,\bbeta_{0})))\}ds\right)\\
	\ge&0.5f_{\min} E\|Q^*(\bbeta)-Q^*(\bbeta_0)\|_2^2 \geq 0.5f_{\min}\lambda_{\min}^{0}\|\Delta\|_{2}^{2}
	\end{split}
	\end{align}
	uniformly for $\{\Delta\in\mathbb{R}^{dp}: \|\Delta\|_2\leq R\}$.
	
	For $1\leq k\leq K$, denote $L_n^{(k)}(\bbeta) =\sum_{i=1}^{n}\rho_{\tau_{k}}\{Y_{i}-Q(\tau_{k},\btheta(\bm X_i,\bbeta))\}$ and $\bar{L}_{k}(\bbeta) =E[n^{-1}L_n^{(k)}(\bbeta)] =E[\rho_{\tau_{k}}\{Y-Q(\tau_{k},\btheta(\bX,\bbeta))\}]$.
	Note that $L_n(\bbeta)=\sum_{k=1}^KL_n^{(k)}(\bbeta)$ and $\bar{L}(\bbeta)=\sum_{k=1}^K\bar{L}_{k}(\bbeta)$.
	Let $\mathcal{E}_{k}^*(\Delta)=|n^{-1}L_n^{(k)}(\bbeta)-n^{-1}L_n^{(k)}(\bbeta_{0})-\{\bar{L}_{k}(\bbeta)-\bar{L}_{k}(\bbeta_{0})\}|$, and we next prove that, uniformly for $r\le \|\Delta\|_{2}\le R$,
	\begin{equation}\label{lemma3-eq2}
	\mathcal{E}_{k}^*(\Delta)\leq C_{\mathcal{E}}\sqrt{\frac{\log p}{n}}\|\Delta\|_{1},
	\end{equation}
	with probability at least $1-\log(\sqrt{dp}r/r_{l})p^{-c^{2}}$ for any $c>1$, where $C_{\mathcal{E}}=C_{E}d2^{d+1}+2L_{Q}C_{X}c$.
	As in Theorem 9.34 in \cite{wainwright2019high}, we use the peeling argument, which is a common strategy in empirical process theory.
	
	\noindent\textit{\textbf{Tail bound for fixed radii:}}
	Define a set $\mathcal{C}(r_{1}):=\left\{\Delta\in\mathbb{R}^{dp}:\|\Delta\|_{1}\le r_{1}\right\}$ for a fixed radii $r_1>0$, and a random variable $A_n(r_1)=r_1^{-1}\sup_{\bbeta\in\mathcal{C}(r_{1})}\mathcal{E}_{k}(\Delta)$. We next show that, for any $t>0$,
	\begin{equation}\label{lemma3-eq3}
	A_n(r_{1})\le C_{E}d2^{d}\sqrt{\frac{\log p}{n}}+L_{Q}C_{X}\sqrt{\frac{t}{n}}.
	\end{equation}
	with probability at least $1-e^{-t}$.

	For $1\leq i\leq n$, denote $\bm{W}_{i}=(Y_i,\bm{X}_i^{\prime})^{\prime}$. Note that random variable $A(r_{1})$ has a form of $f(\bm{W}_{1},\dots,\bm{W}_{n})$, and it is guaranteed by Assumption \ref{Cond-qf} that
	$$\left|f(\bm{W}_{1},\ldots,\bm{W}_{i},\ldots,\bm{W}_{n})-f(\bm{W}_{1},\ldots,\bm{W}_{i\prime},\ldots,\bm{W}_{n})\right|\le n^{-1}L_{Q}C_{X},$$
	i.e., if we replace $\bm W_i$ by $\bm W_{i'}$, while keep other $\bm W_j$ fixed, then $A(r_{1})$ changes by at most $n^{-1}L_Q C_X$.
	As a result, by the bounded differences inequality and for any $t>0$,
	\begin{equation}\label{ineq-bounded-diff}
	A_n(r_{1})\le E[A_n(r_{1})]+L_{Q}C_{X}\sqrt{\frac{t}{n}}
	\end{equation}
	with probability at least $1-e^{-t}$.
	
	In addition, it is implied by Assumption \ref{Cond-qf} that, for all $\bbeta,\widetilde{\bbeta}\in\mathbb{R}^{dp}$,
	\begin{equation*}
	|\rho_{\tau_{k}}\left\{Y_{i}-Q(\tau_{k},\btheta(\bm X_{i},\bbeta))\right\}-\rho_{\tau_{k}}\{Y_{i}-Q(\tau_{k},\btheta(\bm X_{i},\widetilde{\bbeta}))\}|\le L_{Q}\sum_{l=1}^{d}|\bm X_{i}^{\prime}(\bbeta_{l}-\widetilde{\bbeta_{l}})|,
	\end{equation*}
	which leads to
	\begin{equation}
	\begin{split}\label{def-CE}
	E[A_n(r_{1})]\le& \frac{2}{nr_{1}}E\left(\sup_{\bbeta\in\mathcal{C}(r_{1})}\left|\sum_{i=1}^{n}\epsilon_{i}\left[\rho_{\tau_{k}}\{Y_i-Q(\tau_{k},\btheta(\bm X_{i},\bbeta))\}-\rho_{\tau_{k}}\{Y_i-Q(\tau_{k},\btheta(\bm X_{i},\bbeta_0))\}\right]\right|\right)\\
	\le& \frac{C2^{d}}{nr_{1}}E\left(\sup_{\bbeta\in\mathcal{C}(r_{1})}\left|\sum_{l=1}^{d}\sum_{i=1}^{n}V_{il}(\bm X_{i}^{\prime}(\bbeta_{l}-\bbeta_{0l}))\right|\right)
	\le\frac{C2^{d}}{n}E\left(\sum_{l=1}^{d}\left\|\sum_{i=1}^{n}V_{il}\bm X_{i}\right\|_{\infty}\right)\\
	\le&\frac{Cd2^{d}}{n}E\left(\left\|\sum_{i=1}^{n}V_{i1}\bm X_{i}\right\|_{\infty}\right)\le C_{E}d2^{d}\sqrt{\frac{\log p}{n}},
	\end{split}
	\end{equation}
	where $\{\epsilon_{i}\}$ and $\{V_{il}\}$ are $i.i.d.$ Rademacher and standard Gaussian random variables, respectively, the first inequality is due to the symmetrization theorem \cite[Lemma 12]{loh2015regularized}, the second one is by the multivariate contraction theorem \citep[Theorem 16.3]{van2016estimation}, the third one is due to the fact $|\bm X^{\prime}\bbeta|\le \|\bm X\|_{\infty}\|\bbeta\|_{1}$, and the last one is by Lemma 16 of \cite{loh2015regularized} given the sample size of $n\ge c'\log p$ for some $c'>0$.
	The upper bound at \eqref{lemma3-eq3} holds by combining \eqref{ineq-bounded-diff} and \eqref{def-CE}.

	\noindent\textit{\textbf{Extension to uniform radii via peeling:}}
	Define a sequence of sets $\Theta_{l}:=\{\Delta\in\mathbb{R}^{dp}: 2^{l-1}r\le\|\Delta\|_{1}\le 2^{l}r\}$ with $1\leq l\leq N=\log (\sqrt{dp}R/r)$. It can be verified that
	$\{\Delta\in\mathbb{R}^{dp}:r\leq \|\Delta\|_{2}\leq R\} \subseteq\{\Delta\in\mathbb{R}^{dp}:r\leq \|\Delta\|_{1}\leq \sqrt{dp}R\}\subseteq\cup_{l=1}^{N}\Theta_{l}$.
	As a result,
	\begin{align*}
	P&\left(\mathcal{E}_{k}^*(\Delta)\geq C_{\mathcal{E}}\sqrt{\frac{\log p}{n}}\|\Delta\|_{1}, \Delta\in\cup_{l=1}^{N}\Theta_{l} \right)\\
	&\hspace{5mm}\leq \sum_{l=1}^{N}P\left( \mathcal{E}_{k}^*(\Delta)\geq 2^{l-1}rC_{\mathcal{E}}\sqrt{\frac{\log p}{n}}, \Delta \in \Theta_l \right)\\
	&\hspace{5mm}=\sum_{l=1}^{N}P\left( \mathcal{E}_{k}^*(\Delta)\geq (C_{E}d2^{d})(2^{k}r)\sqrt{\frac{\log p}{n}}+L_{Q}C_{X}(2^{k}r)\sqrt{\frac{c^{2} \log p}{n}}, \Delta \in \Theta_l \right)\\
	&\hspace{5mm}\leq \sum_{l=1}^{N}P\left(A(2^{l}r)\ge C_{E}d2^{d}\sqrt{\frac{\log p}{n}}+L_{Q}C_{X}\sqrt{\frac{c^{2}  \log p}{n}}\right),
	\end{align*}	
	where $C_{\mathcal{E}}=C_{E}d2^{d+1}+2L_{Q}C_{X}c$. By applying \eqref{lemma3-eq3}, it holds that
	\[
	P\left(\mathcal{E}_{k}^*(\Delta)\geq C_{\mathcal{E}}\sqrt{\frac{\log p}{n}}\|\Delta\|_{1}, r\leq\|\Delta\|_2\leq R \right) \leq \sum_{k=1}^{N}e^{-c^{2} \log p}=\log(\sqrt{dp}R/r)p^{-c^{2}},
	\]
	i.e. \eqref{lemma3-eq2} holds.
	
	Finally, from \eqref{lemma3-eq1}, \eqref{lemma3-eq2} and Lemma \ref{infty},
	\begin{align*}
	\mathcal{E}_n(\Delta)&\geq \bar{\mathcal{E}}(\Delta) -\sum_{k=1}^K\mathcal{E}_{k}^*(\Delta) -\|n^{-1}\nabla L(\bbeta_{0})\|_{\infty}\|\Delta\|_{1}\\
	&\geq 0.5f_{\min}\lambda_{\min}^{0}\|\Delta\|_{2}^{2}-(KC_{\mathcal{E}}+C_{S})\sqrt{\frac{\log p}{n}}\|\Delta\|_{1}
	\end{align*}	
	uniformly for $r\leq \|\Delta\|_{2}\leq R$ with probability at least $1-c_{1}p^{-c_{2}}-K\log(\sqrt{dp}r/r_{l})p^{-c^{2}}$ for any $c>1$. This accomplishes the proof.
	
\end{proof}

\begin{proof}[Proof of Theorem \ref{high-dim}]
	The proof of this theorem follows from the deterministic analysis in Lemma \ref{deterministic-analysis} and stochastic analysis in Lemmas \ref{infty}, \ref{RSC}.
\end{proof}

\bibliography{refQI}

\newpage

\begin{figure}[ht]   
	\centering
	\subfloat{\includegraphics[width=5cm, height=4cm]{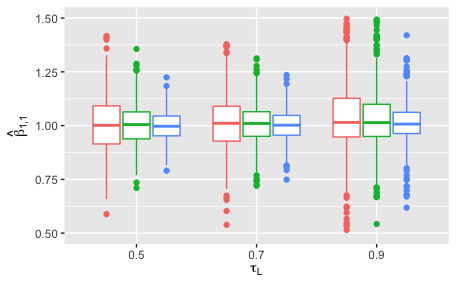}}
	\subfloat{\includegraphics[width=5cm, height=4cm]{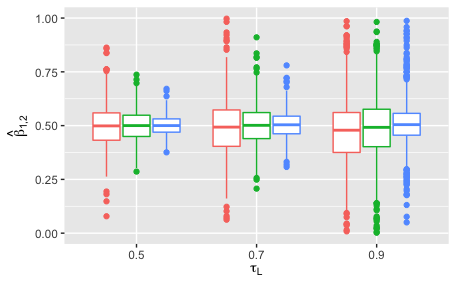}}
	\subfloat{\includegraphics[width=6.5cm, height=4cm]{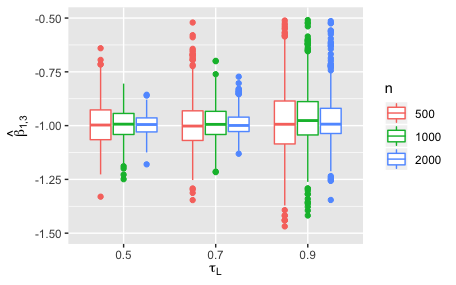}}
	\caption{Boxplots for fitted location parameters of $\widehat{\bbeta}_{1,1}$ (left panel), $\widehat{\bbeta}_{1,2}$ (middle panel), and $\widehat{\bbeta}_{1,3}$ (right panel). Sample size is $n=500$, 1000 or 2000, and the lower bound of quantile range $[\tau_L,\tau_U]$ is $\tau_L=0.5$, 0.7 or 0.9.}
	\label{box:low}
\end{figure}

\begin{figure}[ht]
	\centering
	\includegraphics[width=1\textwidth,height=0.4\textwidth]{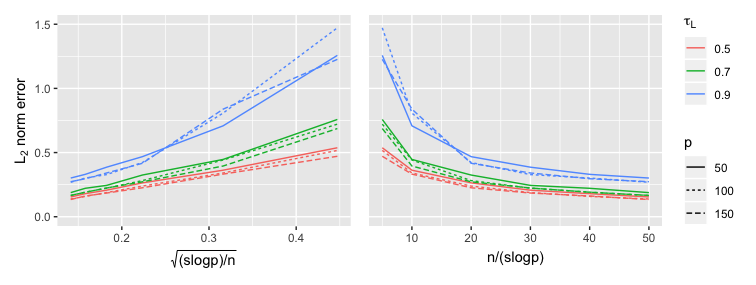}
	\caption{Estimation errors of $\|\widetilde{\bbeta}_n-\bbeta_{0}\|$ against the quantities of $\sqrt{(s\log p)/n}$ (left panel) and $n/{(s\log p)}$ (right panel), respectively.}
	\label{fig:high}
\end{figure}

\begin{table}[ht]
	\centering
	\caption{Mean square errors of the predicted conditonal quantile $Q(\tau^{*},\btheta(\bX,\widehat{\bbeta}_{n}))$ at the level of $\tau^*=0.991$ or $0.995$. The values in bracket refer to the corresponding sample standard deviations of prediction errors in squared loss.}
	\label{tab:low}
	\begin{tabular}{cccccccc}
		\toprule
		&&  &  \multicolumn{2}{c}{$\bm{X}=(1,0.1,-0.2)^T$}  &&    \multicolumn{2}{c}{$\bm{X}=(1,0,0)^T$}  \\ 
		\cmidrule{4-5} \cmidrule{7-8}
		$n$&$[\tau_L,\tau_U]$ &  & 0.991 & 0.995 &  & 0.991 & 0.995 \\ 
		\midrule[1pt]
		& & True & 10.34 & 11.83 &  & 15.13 & 18.84 \\ 		 
		500&$[0.5,0.99]$ &  & 1.32(2.17) & 2.35(4.11) &  & 5.41(11.41) & 12.13(28.19) \\ 
		&$[0.7,0.99]$ &  & 1.42(2.20) & 2.55(4.07) &  & 5.42(10.95) & 12.12(26.73) \\ 
		&$[0.9,0.99]$ &  & 2.00(3.92) & 3.64(7.56) &  & 6.18(12.86) & 14.10(32.78) \\
		\midrule[1pt]
		1000&$[0.5,0.99]$ &  & 0.77(1.68) & 1.39(3.29) &  & 2.67(5.28) & 5.93(12.61) \\ 
		&$[0.7,0.99]$ &  & 0.80(1.39) & 1.44(2.64) &  & 2.62(4.27) & 5.75(9.75) \\
		&$[0.9,0.99]$ &  & 1.31(2.53) & 2.44(5.07) &  & 3.22(5.08) & 7.23(11.78) \\
		\midrule[1pt] 
		2000&$[0.5,0.99]$ &  & 0.32(0.47) & 0.57(0.85) &  & 1.03(1.56) & 2.25(3.49) \\ 
		&$[0.7,0.99]$ &  & 0.36(0.49) & 0.64(0.90) &  & 1.05(1.47) & 2.31(3.24) \\ 
		&$[0.9,0.99]$ &  & 0.70(1.34) & 1.30(2.44) &  & 1.34(1.75) & 3.05(4.06) \\ 
		\bottomrule
	\end{tabular}
\end{table}

\begin{table}[ht]
	\centering
	\caption{Mean square errors of the predicted conditonal quantile $Q(\tau^{*},\btheta(\bX,\widetilde{\bbeta}_{n}))$ at the level of $\tau^*=0.991$ or $0.995$ with $p=50$ and $n=\lfloor ck\log p\rfloor$. The values in bracket refer to the corresponding sample standard deviations of prediction errors in squared loss.}
	\label{tab:high1}
	\begin{tabular}{cccccccc}
		\toprule
		& & &  \multicolumn{2}{c}{$\bm{X}=(1,0.1,-0.2,0,\cdots,0)^T$}  &  & \multicolumn{2}{c}{$\bm{X}=(1,0,0,0,\cdots,0)^T$}   \\ 
		\cmidrule{4-5} \cmidrule{7-8}
		c&$[\tau_L,\tau_U]$ &  & 0.991 & 0.995 &  & 0.991 & 0.995 \\ 
		\midrule[1pt]
		&& True & 10.34 & 11.83 &  & 15.13 & 18.84 \\ 
		10&$[0.5,0.99]$ &  & 1.82(2.61) & 3.23(4.82) &  & 6.75(9.47) & 14.83(21.98) \\ 
		&$[0.7,0.99]$ &  & 2.05(6.00) & 3.78(13.10) &  & 6.11(8.22) & 13.32(18.56) \\
		&$[0.9,0.99]$ &  & 2.92(7.19) & 5.44(15.60) &  & 7.24(10.36) & 15.91(25.05) \\
		\midrule[1pt] 
		30&$[0.5,0.99]$ &  & 0.65(1.59) & 1.15(3.12) &  & 1.97(3.08) & 4.26(6.90) \\ 
		&$[0.7,0.99]$ &  & 0.65(1.49) & 1.18(2.94) &  & 1.96(2.85) & 4.26(6.31) \\ 
		&$[0.9,0.99]$ &  & 0.92(2.15) & 1.66(4.08) &  & 2.22(2.95) & 4.86(6.40) \\
		\midrule[1pt] 
		50&$[0.5,0.99]$ &  & 0.33(0.49) & 0.58(0.84) &  & 1.17(1.77) & 2.52(3.88) \\ 
		&$[0.7,0.99]$ &  & 0.39(0.57) & 0.69(1.01) &  & 1.28(1.93) & 2.78(4.26) \\ 
		&$[0.9,0.99]$ &  & 0.54(0.86) & 0.99(1.58) &  & 1.55(2.39) & 3.51(5.57) \\ 
		\bottomrule
	\end{tabular}
\end{table}

\begin{table}[ht]
	\centering
	\caption{Selection results for regularized estimation with $p=50$ and $n=\lfloor ck\log p\rfloor$. The values in brackets are the corresponding standard errors.}
	\label{tab:high2}
	\begin{tabular}{cccccccc}
		\toprule
		$[\tau_L,\tau_U]$ & c & size & $\textrm{P}_{\textrm{AI}}$ & $\textrm{P}_{\textrm{A}}$ & $\textrm{P}_{\textrm{I}}$ & FP & FN \\
		\midrule[1pt]
		$[0.5,0.99]$	& 10 & 9.04(0.99) & 91.6 & 96 & 95.6 & 0.06(0.68) & 0.47(2.34) \\ 
		& 30 & 9.00(0.00) & 100 & 100 & 100 & 0.00(0.00) & 0.00(0.00) \\ 
		& 50 & 9.00(0.00) & 100 & 100 & 100 & 0.00(0.00) & 0.00(0.00) \\ 
		\midrule[1pt]
		$[0.7,0.99]$	& 10 & 8.91(1.25) & 79 & 82.6 & 95.4 & 0.07(0.84) & 2.02(4.52) \\ 
		& 30 & 9.00(0.08) & 99.4 & 99.6 & 99.8 & 0.00(0.03) & 0.04(0.70) \\
		& 50 & 9.00(0.00) & 100 & 100 & 100 & 0.00(0.00) & 0.00(0.00) \\
		\midrule[1pt]
		$[0.9,0.99]$	& 10 & 8.56(0.99) & 48.4 & 54.6 & 90.4 & 0.10(0.41) & 6.51(8.43) \\
		& 30 & 8.88(0.38) & 87.8 & 88.4 & 99.2 & 0.01(0.06) & 1.42(4.10) \\
		& 50 & 8.96(0.23) & 96.4 & 96.6 & 99.8 & 0.00(0.03) & 0.44(2.50) \\
		\bottomrule
	\end{tabular}
\end{table}

\begin{figure}[ht]
	\centering
	\subfloat{\includegraphics[width=7cm,height=5cm]{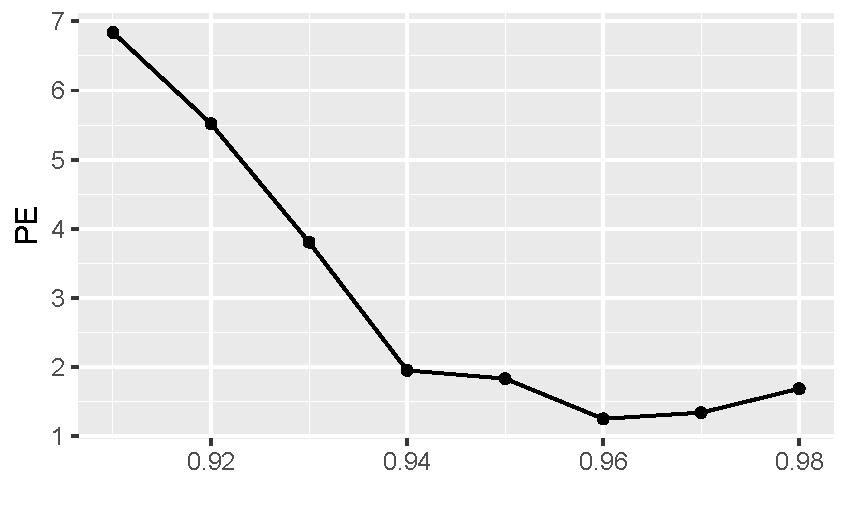}}
	\subfloat{\includegraphics[width=10cm,height=5cm]{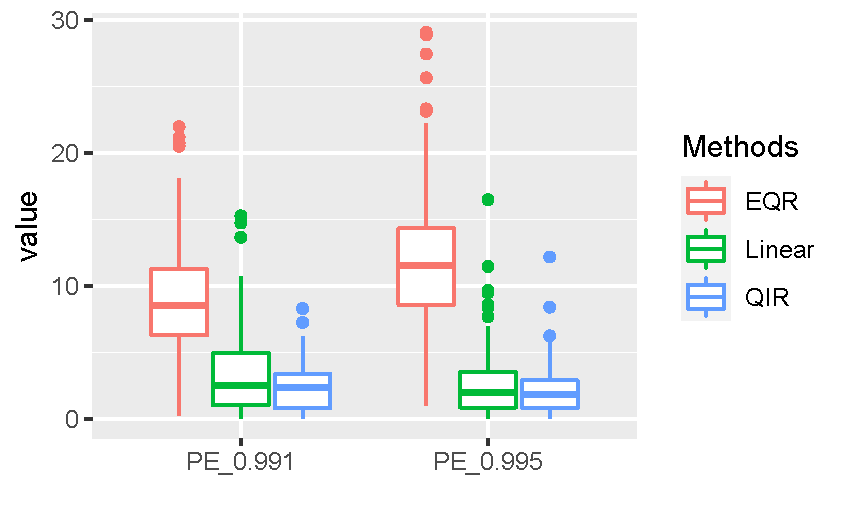}}
	\caption{Plot of PEs against $\tau_L$ (left panel) and boxplots of PEs from the extreme quantile regression (EQR), linear quantile regression (Linear) and QIR models
		at two target levels of $\tau^*=0.991$ and 0.995 (right panel).}\label{fig:compare}
\end{figure}

\begin{figure}[ht]
	\centering
	\includegraphics[height=5cm, width=7cm]{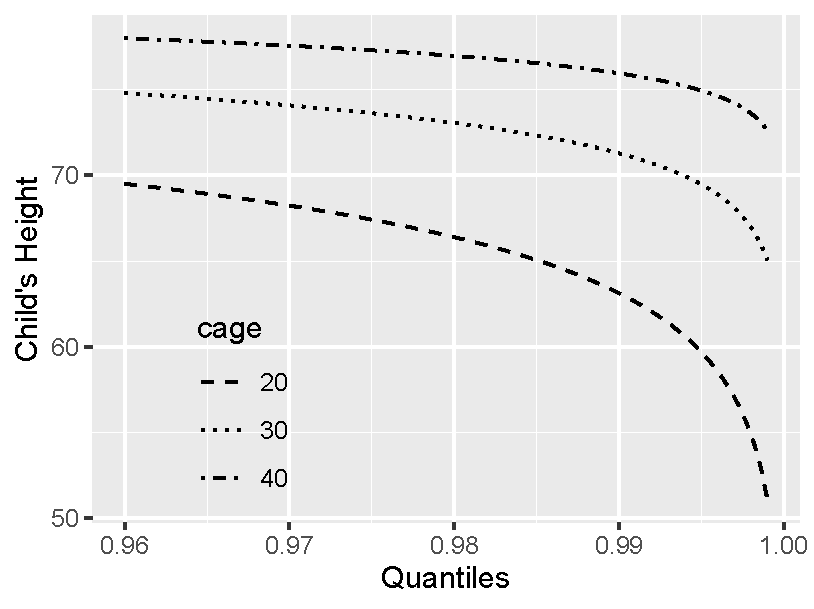} 
	\includegraphics[height=5cm, width=7cm]{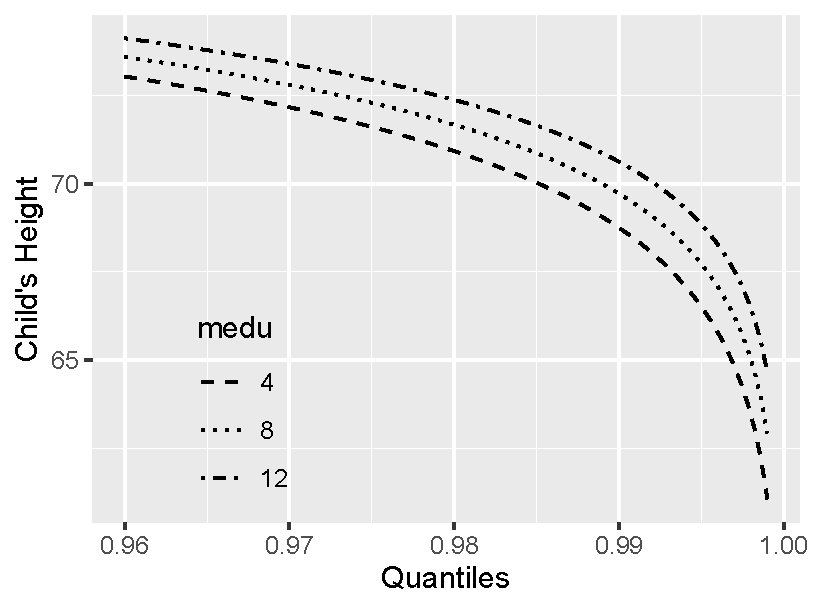}\\
	\includegraphics[height=5cm, width=7cm]{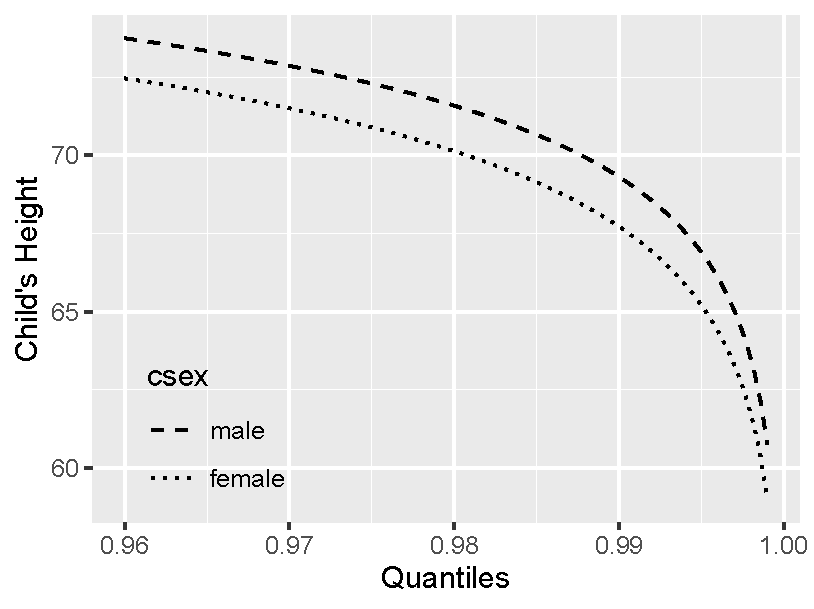} 
	\includegraphics[height=5cm, width=7cm]{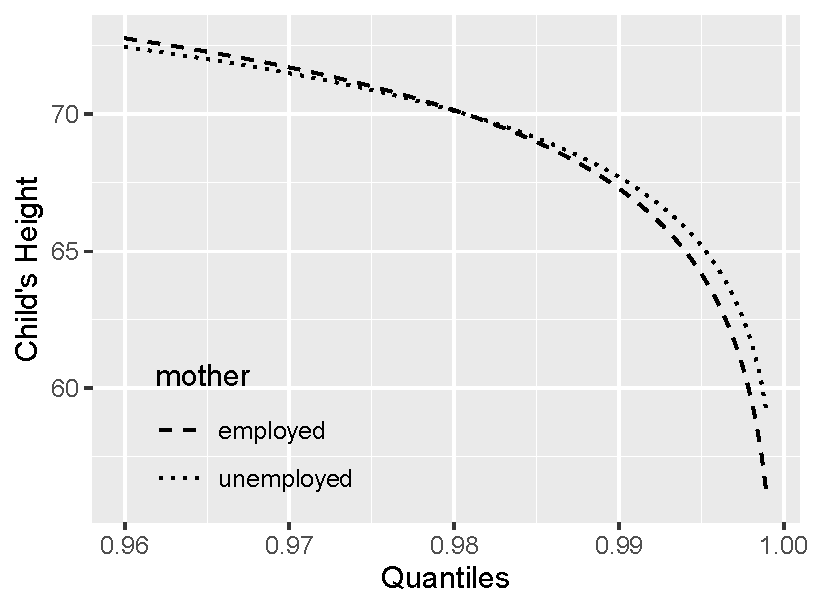}
	\\
	\caption{Quantile curves for child's age in months (top left), mother's education in years (top right) on the three target quantiles. Effects of child's sex (bottom left) and mother's unemployment condition (bottom right).}\label{fig:interpretation}
\end{figure}

\end{document}